\documentclass[journal,twoside,print]{ieeecolor}

\usepackage{cite}
\usepackage{amsmath,amssymb,amsfonts}
\usepackage{algorithmic}
\usepackage{graphicx}
\usepackage{textcomp}

\usepackage{array}
\usepackage{mdwmath}
\usepackage{eqparbox}
\usepackage{url}
\usepackage{amssymb}
\usepackage{booktabs}
\usepackage{enumerate}
\usepackage{balance}

\usepackage{tikz}
\usetikzlibrary{calc, shapes, arrows, positioning}

\definecolor{subsectioncolor}{rgb}{0,0,0}

\newtheorem{theorem}{Theorem}

\newtheorem{remark}{Remark}

\newtheorem{lem}{Lemma}
\newtheorem{definition}{Definition}
\newtheorem{assum}{Assumption}

\begin{document}
%

\twocolumn[  
    \begin{@twocolumnfalse}
        \begin{center}

            \huge Bilateral Peer-to-Peer Energy Trading via Coalitional Games\\
            \vspace{5mm}
            \large Aitazaz Ali Raja and Sergio Grammatico

         \end{center}
     \end{@twocolumnfalse}
]

%


\vspace{15mm}
\begin{abstract}
In this paper, we propose a bilateral peer-to-peer (P2P) energy trading scheme under single-contract and multi-contract market setups, both as an assignment game, a special class of coalitional games.  {The proposed market formulation allows for efficient computation of a market equilibrium while keeping the desired economic properties offered by the coalitional games. Furthermore, our market model allows buyers to have heterogeneous preferences (product differentiation) over the energy sellers,  which can be economic, social, or environmental. To address the problem of scalability in coalitional games, we design a novel distributed negotiation mechanism that utilizes the geometric structure of the equilibrium solution to improve the convergence speed. Our algorithm enables market participants (prosumers) to reach a consensus on a set of ``stable" and ``fair" bilateral contracts which encourages prosumer participation.} The negotiation process is executed with virtually minimal information requirements on a time-varying communication network that in turn preserves privacy. We use operator-theoretic tools to rigorously prove its convergence.  {Numerical simulations illustrate the benefits of our negotiation protocol and show that the average execution time of a negotiation step is much faster than the benchmark.}

\end{abstract}


\let\thefootnote\relax\footnotetext{Aitazaz Ali Raja and Sergio Grammatico are with Delft Center for Systems and Control, TU Delft, The Netherlands. (e-mail addresses: a.a.raja@tudelft.nl; s.grammatico@tudelft.nl). This work was partially supported by NWO under research project P2P-TALES (grant n. 647.003.003) and the ERC under research project COSMOS, (802348).}

\vspace{-0.1mm}
\section{Introduction}

\IEEEPARstart{M}{odernization} of power systems is rapidly materializing under the smart grid framework. A major part of this transformation is taking place on the consumer side, due to the increasing penetration of distributed energy resources (DER) along with the deployment of communication and control technologies. 
These technologies  enable consumers to have an active interaction with the grid by an informed control over their energy behavior, thus they are referred as ``prosumers". \\
To realize their full potential, prosumers should engage more actively with energy markets. Currently the direct participation of prosumers in the whole sale energy market is technically and economically non-viable. Hence, small-scale prosumers interact with aggregating entities such as retailers to deliver their excess energy to the grid \cite{morstyn2018using}. Retailers  usually offer a considerably lower price for the energy sold by prosumers, e.g. feed-in-tariff (FiT), compared to the buying price that they charge \cite{han2018incentivizing}. 
To ensure an economically appealing role of prosumers, peer-to-peer (P2P) energy trading represents a disruptive demand side energy management strategy \cite{Tushar2020}. In fact, P2P markets enable prosumers to locally exchange energy on their own terms of transactions. This direct control over trading allows prosumers to make profitable interactions, thus it encourages wider participation \cite{tushar2018peer}. Furthermore, such a local exchange of energy at the demand side also provides significant benefits to the system operators for example in terms of peak shaving \cite{moret2018energy},
lower investments in grid capacity 
and improvement in overall system reliability \cite{morstyn2018using}. \\
However, there are strong mathematical challenges in designing a comprehensive P2P energy market mechanism which seeks a market equilibrium while incorporating a self-interested decision-making attitude by the participants \cite{tushar2018transforming}. Despite its mathematical sophistication, the mechanisms need to be easily interpretable for the participation of laypersons, e.g. residential prosumers. Along this direction, researchers have recently presented several interesting formulations. In the literature, P2P energy markets are proposed under various architectures that can be broadly categorized as centralized markets (community-based trading), decentralized markets (bilateral trading) and combinations there of. The features of centralized market architecture include an indirect interaction of market participants via assisting platforms, no negotiatory role for market participants and single market wide energy trading price, evaluated centrally. Among others, in \cite{moret2018energy} Moret and Pinson present a centralized local energy market where groups of prosumers (energy collectives) interact with each other and with the system operator via a community manager to make energy exchanges. In \cite{morstyn2018multiclass}, Morstyn and McCulloch treat energy as a heterogeneous product that can be differentiated based on the attributes of its source. Centralization is achieved by a platform agent that is supposed to maximize social welfare by setting prices and that enables energy exchanges among prosumers and with the wholesale electricity market. In \cite{Crespo-Vazquez2020}, Vazquez, Al-Skaif and Rueda present a community based market that models the decision making into three sequential steps, solved using distributed optimization, where the energy is exchanged via a local pool at single clearing price. \\
Decentralized P2P markets can allow for direct buyer-seller (bilateral) interaction with possibly different energy trade price for each bilateral contract. In \cite{Sorin2019}, Sorin, Bobo and Pinson formulate a decentralized P2P market architecture based on a multi-bilateral economic dispatch with a possibility of product differentiation, where the solution is obtained by solving a distributed optimization problem. Another decentralized P2P market is presented by Morstyn, Teytelboym and McCulloch in \cite{Morstyn2019}, which is formulated as a matching market that seeks a stable bilateral contract network. In both works, prosumers are allowed to make bilateral contracts, i.e., each energy transaction can take place at a different price.
In \cite{nguyen2020optimal}, Nguyen also presents a decentralized P2P market with a clearing mechanism based on the alternating direction method of multipliers.\\
{Within the industrial informatics community, P2P energy platforms have received strong research attention under both centralized and decentralized architectures. Recently, in \cite{ullah2021two}, the authors propose a two-tier market corresponding to inter and intra-region interactions where a DSO acts as a representative for each region and the price of energy trade between regions is evaluated centrally. In \cite{cui2019efficient}, the authors also present a two-level market for trading energy with and within energy communities. Each community is represented by an aggregator that decides inter-community trading price whereas intra-community trading is done at a fixed price. Both the works propose distributed optimization based market solutions and don't allow for bilateral economic interaction on prosumer (peer) level. A decentralized P2P market is analysed in \cite{alskaif2021blockchain} with the possibility of a bilateral trade however, their focus is on the trading preferences of prosumers rather than the mechanism design. The blockchain based implementations to address privacy and security in P2P platforms are also addressed \cite{pradhan2021flexible}.
Most of the works reviewed above lack any discussion on the economic properties of the proposed trading strategies, which are critical for the practicality of any market mechanism. Next, we build a case for our proposal of a P2P market mechanism that allows for a decentralized (fully P2P) interaction and also ensures desirable economic properties.} \\ 
{\textit{Market design}: The key desirable properties of electricity market design are market efficiency, incentive compatibility, cost recovery and revenue adequacy \cite{schweppe2013spot}. Unfortunately, by Hurwicz's impossibility theorem, 
no market mechanism can satisfy all four properties simultaneously and a trade‐off has to be found. Centralized electricity markets are usually cleared based on the locational marginal pricing, which only satisfies cost recovery and revenue adequacy. Similarly, the VCG mechanism, utilized in several P2P market designs satisfies market efficiency, incentive compatibility, and cost recovery making financial deficit possible for the market operator. In our context of P2P market design where the participants are relatively small prosumers, hence are not generally capable of exercising market power, we can reasonably assume them to be truthful, thus enforcing incentive compatibility. The other market properties can be satisfied by the solutions of canonical coalitional games, thus they provide the required mathematical foundation for the design of P2P markets. We refer the reader to \cite{Saad2009} for details about the classes (canonical, coalition formation, and coalitional graph) of coalitional games. In the general setup of coalitional games, the solutions with the required properties, such as the core, suffer the issue of scalability as all possible sub-coalitions, i.e, $2^N - 1$ for $N$ agents, must be considered. Based on these considerations, here we model the P2P electricity market as an assignment game, a special class of coalitional games, for which the solution requires the information about coalition pairs only. This solves the scalability issue while keeping the desired market properties \cite{shapley1971assignment}. Furthermore, our model also allows for \textit{bilateral} interactions where buyers can exercise their preferences over the sellers, as well as their energy sources which, in our opinion,  captures the true spirit of P2P trading. Specifically, in this paper, we propose an easily interpretable decentralized (bilateral) P2P energy market that allows for a heterogeneous treatment of energy by utilizing concepts from coalitional game theory for mechanism design and for proving the plausibility of the equilibrium solution.}\\
 {Coalitional game theory provides rigorous analytical tools for the cooperative interactions among agents with selfish interests, and in fact it has received a strong attention from smart-grid researchers recently.} For instance, the authors in \cite{tushar2018peer} propose a P2P energy trading scheme in which prosumers form a coalition to trade energy among themselves at a (centralized) mid-market rate which in turn ensures the stability of the coalition  \cite{tushar2019motivational}.  
In \cite{tushar2019grid}, the authors formulate a coalition formation game for P2P energy exchange among prosumers, and the resulting coalition structure is shown to be stable.
The price of exchange is determined by a central auctioneer based on a double auction mechanism. Another coalition formation game is presented in \cite{tushar2020coalition}, which allows prosumers to optimize their battery usage for P2P energy trading. The outcome is shown to be stable, optimal and prosumer-centric. The authors in \cite{luo2018distributed} propose decentralized bilateral negotiation among prosumers for energy exchange via coalition formation, but without considering coalitional-game-theoretic stability.  {Coalition formation games have also been utilized for cooperative charging of electric vehicles (EVs). 
The authors in \cite{yu2014phev} consider EVs in different regions with different discharging prices. Based on this difference, EVs form a coalition structure to exchange energy. Then in \cite{wang2021blockchain}, the authors formulate a coalition formation game among private charging piles to optimally provide charging services to EVs by sharing their resources. In both works the considered games are non-super-additive, thus they employee a merge-and-split protocol to reach a stable coalition structure. The protocol does not determine the price of exchange but only a stable match. In Table \ref{tab: literature}, we provide a comparison between the features of our P2P market model and those of the most relevant literature. We note that in Table \ref{tab: literature}, stability is considered in the context of coalitional game theory and we mark the presence of guarantees on the market properties only if they are explicitly discussed in the paper.} 

To the best of our knowledge, the literature on coalitional game theoretic formulation of P2P markets lacks a development of bilateral P2P model via most widely studied and easily interpretable class of coalitional games, i.e., canonical coalitional games. Along with the mathematical rigor provided by the game formulation, its straight forward interpretation is also an important feature for a P2P market design, intended to encourage the participation of small prosumers with low technical knowledge.  {Though canonical games seem the most natural approach to model bilateral market, the hindrance in its adoption comes from high computational complexity and coalition stabilizing contract prices might not exist, i.e., the core set might be empty.} To address these, here, we model P2P energy trading as a canonical coalitional game that allows for bilateral energy trading contracts and guarantees the existence of \textit{stable} contract prices that represent a competitive equilibrium of the market. Furthermore, the negotiation mechanism enables market participants for an efficient and convenient settlement on the \textit{stable} and \textit{fair} contract prices.\\
\begin{table}[]
 \centering
\caption{Comparison with the state-of-the-art employing coalitional game theory}
\label{tab: literature}
\begin{tabular}{lccccc} \toprule
\textbf{Features} & \cite{tushar2018peer} & \cite{tushar2019grid} & \cite{tushar2020coalition}  & \cite{luo2018distributed} & This paper \\ \midrule
Bilateral contracts&     $\times$ & $\times$   & $\times$ & \checkmark & \checkmark\\ \midrule
Guarantees on\\ market properties&        \checkmark  & $\times$   & $\times$ & $\times$   & \checkmark\\ \midrule
Product differentiation& $\times$ &  $\times$ & $\times$ & $\times$   & \checkmark \\ \midrule
Distributed computation& $\times$ &  $\times$ & $\times$ & \checkmark   & \checkmark \\ \midrule
Stability of market& \checkmark & \checkmark & \checkmark & $\times$   &\checkmark\\ \bottomrule
 \end{tabular}
\vspace{-0.1mm}
 \end{table}
\textit{Contribution}: 
Our key contributions are summarized next:
\vspace{-0.1mm}
\begin{itemize}
    \item We formulate P2P energy trading as a bilateral assignment game (coalitional game), which is easily interpretable and allows for product differentiation to accommodate the heterogeneous preferences of buyers. This novel formulation ensures the existence of a ``stable" set of bilateral contracts that is an equilibrium (Section \ref{sec: Peer-to-Peer market as coalitional game}). Furthermore, our market formulation ensures the desirable economic properties of the mechanism, which are market efficiency, cost recovery, and revenue adequacy;
    \item We develop \textit{single-contract} and \textit{multi-contract} setups of bilateral P2P energy market with different computational burdens and features (Section \ref{subsec: single-contract} and \ref{subsec: multi contract});
     \item We develop a novel distributed negotiation mechanism presented as a fixed-point iteration where buyers-sellers communicate locally over a possibly time-varying communication network. We exploit the geometrical structure of the \textit{core} solution together with operator theory to formulate our algorithm via linear operations, thus considerably reducing the computational complexity of the negotiation, strongly improving over \cite{raja2021}, \cite{bauso2015distributed}. We show that the mechanism converges to a payoff allocation in the core of the assignment game (Section \ref{sec: solution mechanism});
    \item  We present our algorithm in a generalized form which enables fast convergence and allows participants to negotiate for ``fair" contracts in the interior of the core set \cite{shapley1971assignment}. The level of information requirement in our mechanism preserves privacy among the market participants. Furthermore, our algorithm is based on consensus protocols, which are easier to analyze and embed on real hardware, instead of dual variables (e.g. in \cite{nguyen2020optimal}), to reach a common price vector among the participants.
\end{itemize}
 {\noindent \textit{Notation and definitions}: Given a mapping $M: \mathbb{R}^n \rightarrow \mathbb{R}^n, \mathrm{fix}(M):= \{x \in \mathbb{R}^n \mid x = M(x)\} $ denotes the set of its fixed points. $\text{Id}$ denotes the identity operator. For a closed set \(C \subseteq \mathbb{R}^{n},\) the mapping $\mathrm{proj}_C$: \(\mathbb{R}^{n} \rightarrow C\) denotes the projection onto \(C,\) i.e., \(\operatorname{proj}_C(x)=\) \(\arg \min _{y \in C}\|y-x\| .\) An over-projection operator is denoted by $\mathrm{overproj}_C : = 2\mathrm{proj}_C - \text{Id} $. For a set $S$ the power set is denoted by $2^S$. \(A \otimes B\) denotes the Kronecker product between the matrices \(A\) and \(B .\) $I_N$ denotes an identity matrix of dimension $N \times N$. {For $x_{1}, \ldots, x_{N} \in \mathbb{R}^{n},$ $\mathrm{col}(\left(x_{i}\right)_{i \in(1, \ldots, N)}):=\left[x_{1}^{\top}, \ldots, x_{N}^{\top}\right]^{\top}.$} 
$\mathrm{dist}(x,C)$ denotes the distance of $x$ from a closed set \(C \subseteq \mathbb{R}^{n},\) i.e., $\mathrm{dist}(x,C):= \mathrm{inf}_{y \in C} \|y-x\|$. {For a closed set \(C \subseteq \mathbb{R}^{n}\) and $N \in \mathbb{N}, C^N:=\prod_{i=1}^{N} C_{i}$}. A continuous mapping $M : \mathbb{R}^n \rightarrow \mathbb{R}^n$ is a paracontraction, with respect to a norm $\|\cdot\|$ on $\mathbb{R}^n$, if $ \|M(x) - y\| < \|x - y\|,$ for all $(x,y) \in (\mathbb{R}^n \backslash \mathrm{fix}(M)) \times \mathrm{fix}(M)$.}
\vspace{-0.1mm}
\section{Background on Coalitional Games and Assignment Games}\label{sec: background on CG}
\vspace{-0.1mm}
\subsection{Coalitional games}\label{subsec: CG}
Let us first provide some mathematical background on coalitional game theory and then describe assignment games, a special class of coalitional games. A coalitional game consists of a set of agents, indexed by $\mathcal{I} = \{1, \ldots, N\}$, who cooperate to receive a higher individual return compared to that due to non-cooperative actions. The utility generated by this cooperation is defined by a value function $v$.
\begin{definition} [Coalitional game]
Let $\mathcal{I} = \{1, \ldots, N\}$ be a set of agents. A coalitional game is a pair $\mathcal{G} = (\mathcal{I}, v)$ where $v: 2^\mathcal{I} \to \mathbb{R}$ is a value function that assigns a real value, $v(S)$, to each coalition $S \subseteq \mathcal{I}$. $v(\mathcal{I})$ is the value of so-called grand coalition. By convention, $v(\varnothing) = 0$. $\hfill \square$
\end{definition} 
In a coalitional game, the value generated by a coalition $S$, $v(S)$, is distributed among the members of $S$ as a payoff. For each $i \in S$, the element $x_i$ of a payoff vector $\boldsymbol{x} \in \mathbb{R}^{|S|}$ represents the share of agent $i$ of the value $v  (S)$.
 {For a game with a grand coalition $\mathcal{I}$
 we assume that each agent $i  \in \mathcal{I}$ is rational and demands an efficient payoff vector.} Mathematically, this means that the payoff vector proposed by each agent must belong to its \textit{bounding set}.
\begin{definition}[Bounding set]\label{def: bounding set}
For a coalitional game $\mathcal{G}=(\mathcal{I}, v)$, the set
\begin{equation} \label{eq: bounding set}
\begin{array}{ll}
\vspace{-0.1mm}
\mathcal{X}_i :=    \bigg\{x \in \mathbb{R}^N \mid & \sum_{j \in \mathcal{I}} x_j = v (\mathcal{I}),\\
  &  \sum_{j \in S} x_j \geq v  (S), \forall S \subset \mathcal{I} \text{ s.t. } i \in S \bigg\}
\end{array}
\end{equation}
denotes the bounding set of an agent $i \in S$. $\hfill \square$
\end{definition}
We note from (\ref{eq: bounding set}) that Bounding half space is closed and convex, a polytope with special geometry, thus we can represent the bounding set as the intersection of \textit{bounding half-spaces}.

\begin{definition} [Bounding half-spaces]\label{def: hyperplanes}
For a coalitional game $(\mathcal{I},v)$ and a coalition $S \subset \mathcal{I}$ the bounding half-space is a set $  H(S) := \textstyle \{ x \in \mathbb{R}^N \mid \sum_{i \in S} x_i \geq v(S) \}. $
Moreover, let the set
\begin{equation}\label{eq: hyperplanes}
\mathcal{H}_i :=  \{ H(S) \mid S \subset \mathcal{I}, i \in S  \}
\end{equation} denote the set of all bounding half-spaces corresponding to the set of rational and efficient payoffs for an agent $i$, i.e., the bounding set $\mathcal{X}_i$ in (\ref{eq: bounding set}). $\hfill \square$
\end{definition}
Now, using half-spaces as in Definition \ref{def: hyperplanes}, we can write the bounding set as $\mathcal{X}_i = \bigcap_{S \subseteq \mathcal{I}| i \in S} H(S) $.\\
Since a rational agent $i$ agrees only on a payoff in its bounding set $\mathcal{X}_i$ thus, a mutually agreed payoff shall belong to the intersection of the bounding sets of all the agents. Interestingly, this intersection corresponds to the core, the solution concept that relates to the 
stability of a grand coalition \cite{Saad2009}.
\begin{definition} (Core):
The core $\mathcal{C}$ of a coalitional game ($\mathcal{I},v$) is the following set of payoff vectors:
\begin{equation} \label{core}
  \mathcal{C} :=  \{ x \in \mathbb{R}^N \mid \sum_{i \in \mathcal{I}} x_i = v(\mathcal{I}),\sum_{i \in S} x_i \geq v(S), \forall S \subseteq \mathcal{I}  \},  
\end{equation}
where, the term $\sum_{i \in \mathcal{I}}x_i = v(\mathcal{I})$ ensures the efficiency and $\sum_{i \in S}x_i \geq v(S)$ shows the rationality of  a payoff. $\hfill \square$
\end{definition}
 {We note from (\ref{core}) that the core set $\mathcal{C}$ is closed and convex, a polytope with special geometry. These facts allow us to represent the core by using the bounding sets, as in Definition \ref{def: bounding set}: We can in fact write the core as $\mathcal{C} = \bigcap_{i=1}^N \mathcal{X}_i$.}\\
In the sequel, we deal with the grand coalition only, therefore we use the core $\mathcal{C}$ as the solution concept. Next, we give some background on a subclass of transferable utility (TU) coalitional games called assignment games.
\vspace{-0.1mm}
\subsection{Assignment games}\label{subsec: assignment game}
An assignment game models a bilateral one-to-one matching market with the primary objective of finding optimal assignments between the two sides, for example, matching buyers to sellers \cite{shapley1971assignment}. Thus,  let us refer to the sets of agents on the two sides of the market as buyers and sellers and denote them by $\mathcal{I}_\mathcal{B}$ and $\mathcal{I}_\mathcal{S}$, respectively. Here, each seller $j \in \mathcal{I}_\mathcal{S}$ owns a good for which declares the value of at least $c_j$; whereas, for each buyer $i \in \mathcal{I}_\mathcal{B}$, the ceiling worth of the good of seller $j$ is $h_{i,j}$.  Then, the value function that gives value to a simplest \textit{meaningful} coalition, i.e., a buyer-seller pair, reads as:
\begin{equation}\label{eq: value function}
  v(i,j) = \max\{0,h_{i,j} - c_j\}. 
\end{equation}
Here, with a slight abuse of notation, we refer to $v(\{i,j\})$ by $v(i,j)$.
We note that any assignment which is favorable to both parties  must satisfy $h_{i,j} > c_i$. Furthermore,  one-sided coalitions generate no value, i.e., $v(S) = 0$ if $S \subseteq \mathcal{I}_\mathcal{B}$ or $S \subseteq \mathcal{I}_\mathcal{S}$, thus only mixed coalitions are meaningful.\\
 Interestingly, the buyer-seller pairs alone suffice to determine the market completely. Using this observation, we define an assignment matrix $M = [v(i,j)]$ for all pairs $(i,j) \in \mathcal{I}_\mathcal{B}\times \mathcal{I}_\mathcal{S}$.

\begin{definition} [Value function]\label{def: value function}
Let $\mathcal{I}_\mathcal{B} = \{1, \ldots, N_\mathcal{B}\}$ and $\mathcal{I}_\mathcal{S}= \{1, \ldots, N_\mathcal{S}\}$ be the sets of buyers and sellers, respectively. Let $M = [v(i,j)]_{(i,j) \in \mathcal{I}_\mathcal{B}\times \mathcal{I}_\mathcal{S}}$ be an assignment matrix with $v(i,j)$ as in (\ref{eq: value function}). Given $\mathcal{B} \subseteq \mathcal{I}_\mathcal{B}$ and $\mathcal{S} \subseteq \mathcal{I}_\mathcal{S}$, let $\mathcal{P}(\mathcal{B}, \mathcal{S})$ be the set of all possible matching configurations between $\mathcal{B}$ and $\mathcal{S}$, where a matching configuration is a set of two-sided matchings such that a seller (buyer) is matched with at most one buyer (seller). Then, the value function $v_{M}: \mathcal{I}_\mathcal{B} \cup \mathcal{I}_\mathcal{S} \to \mathbb{R}$ is defined as, $v_{M}(\mathcal{B} \cup \mathcal{S})=\underset{P \in \mathcal{P}(\mathcal{B}, \mathcal{S})}{\max} \sum_{(i,j) \in P} v(i,j).$ $\hfill \square$
\end{definition} 
Let us now formally define an assignment game.
 \begin{definition} [Assignment game]\label{def: assignment game}
Let $\mathcal{I}_\mathcal{B} = \{1, \ldots, N_\mathcal{B}\}$ and $\mathcal{I}_\mathcal{S}= \{1, \ldots, N_\mathcal{S}\}$ be the sets of buyers and sellers, respectively.  An assignment  game is a pair $\mathcal{M} = (\mathcal{I}_\mathcal{B} \cup \mathcal{I}_\mathcal{S}, v_M)$, where the value function $v_M$ is as in Definition \ref{def: value function}. 
 $\hfill \square$
\end{definition} 
Next, we reformulate the core in (\ref{core}) for assignment games.
\begin{definition} [Core of assignment game]
The core $\mathcal{C}_M$ of an assignment game $(\mathcal{I}_\mathcal{B} \cup \mathcal{I}_\mathcal{S}, v_M)$ is the following set:
\begin{align}\label{core assignment}
 & \mathcal{C}_M :=  \{ (x', x'') \in \mathbb{R}^{N_\mathcal{B}} \times \mathbb{R}^{N_\mathcal{S}} \mid \textstyle \sum_{i \in \mathcal{I}_\mathcal{B}} x_i' + \sum_{j \in \mathcal{I}_\mathcal{S}} x''_j = \nonumber\\
  & \qquad v_M(\mathcal{I}_\mathcal{B} \cup \mathcal{I}_\mathcal{S}), x_i' + x''_j \geq v(i,j) \text{ for all } (i,j) \in \mathcal{I}_\mathcal{B} \times \mathcal{I}_\mathcal{S}  \}. 
\end{align}
\end{definition}

\begin{remark}[Non-emptiness of core {\cite{shapley1971assignment}}]\label{lem: non empty core}
An assignment game (as in Definition \ref{def: assignment game}) has a non-empty core. $\hfill \square$
\end{remark}
We note that the core of an assignment game is defined by two sided pair coalitions instead of all possible coalitions in (\ref{core}), which considerably reduces the complexity of solving an assignment game. Thus, an assignment game  presents a more practical approach, compared to the general coalitional game theory, towards formulating a bilateral P2P market.  \\
For an optimally matched pair $(i,j) \in \mathcal{I}_\mathcal{B} \times \mathcal{I}_\mathcal{S}$, the payoff $(x'_i, x''_j)$ determines the contract price $\lambda_{i,j}$. In a bilateral trade, buyer $i$ pays to seller $j$ the difference of the price they initially offered and his payoff, i.e.,  $\lambda_{i,j} = h_{i,j} - x'_i$. For brevity, in the sequel, we use the collective payoff vector for buyers and sellers, i.e.,  $x = \mathrm{col}(x', x'')$, where $x' \in \mathcal{I}_\mathcal{B}$ and $x'' \in \mathcal{I}_\mathcal{S}$.\\
We remark that, for P2P markets modelled as coalitional games it is not plausible to adopt centralized methods for computation of a payoff in the core because the  core  set  is  not  singleton and different core payoffs can favor different sides (buyers or sellers) of the market. In fact, each core set has a buyer optimal and a seller optimal point as the two extremes \cite{shapley1971assignment}. Thus, it raises the possibility of biased behavior of the central operator which may jeopardize the confidence of market participants. Furthermore, in practice, bilateral agreements should be directly negotiated by the self-interested agents. Thus, in Section \ref{sec: solution mechanism} we propose a distributed solution mechanism in which the agents negotiate to autonomously reach a mutual agreement, i.e., consensus on the payoff vector and consequently on the trading prices. Finally, let us formalize the notion of consensus set.
\begin{definition} [Consensus set]\label{def: consensus set}
The consensus set $\mathcal{A} \subset \mathbb{R}^{N^{2}}$ is defined as:
\begin{equation}\label{eq: consensus}
\mathcal{A} := \{\mathrm{col}(\boldsymbol{x}_1, \ldots, \boldsymbol{x}_N) \in \mathbb{R}^{N^{2}} \mid \boldsymbol{x}_i = \boldsymbol{x}_j, \forall i,j \in \mathcal{I}\}. 
\end{equation} $\hfill \square$
\end{definition}
\vspace{-0.1mm}
\section{P2P market as an assignment game}\label{sec: Peer-to-Peer market as coalitional game}
\subsection{Modelling}\label{subsec: P2P modelling}
  In this section, we present two setups of a bilateral P2P energy market as assignment games namely, single-contract market and multi-contract market. The participants of the market are partitioned into buyers and sellers where, a seller is a prosumer who owns an energy source including renewable (RES) and/or energy storage (ES) with an excess energy available, for a trading period, while a buyer can be a mere consumer as well. The market is operated by a central coordinator (market operator) who has complete information of buying and selling bids, and is also responsible for maximizing the overall market welfare.  {In Figure \ref{fig: market structure}, we illustrate the high level concept of the proposed P2P energy market structure.}
 
\begin{figure}[t]
\centering
\includegraphics[width =0.9\linewidth]{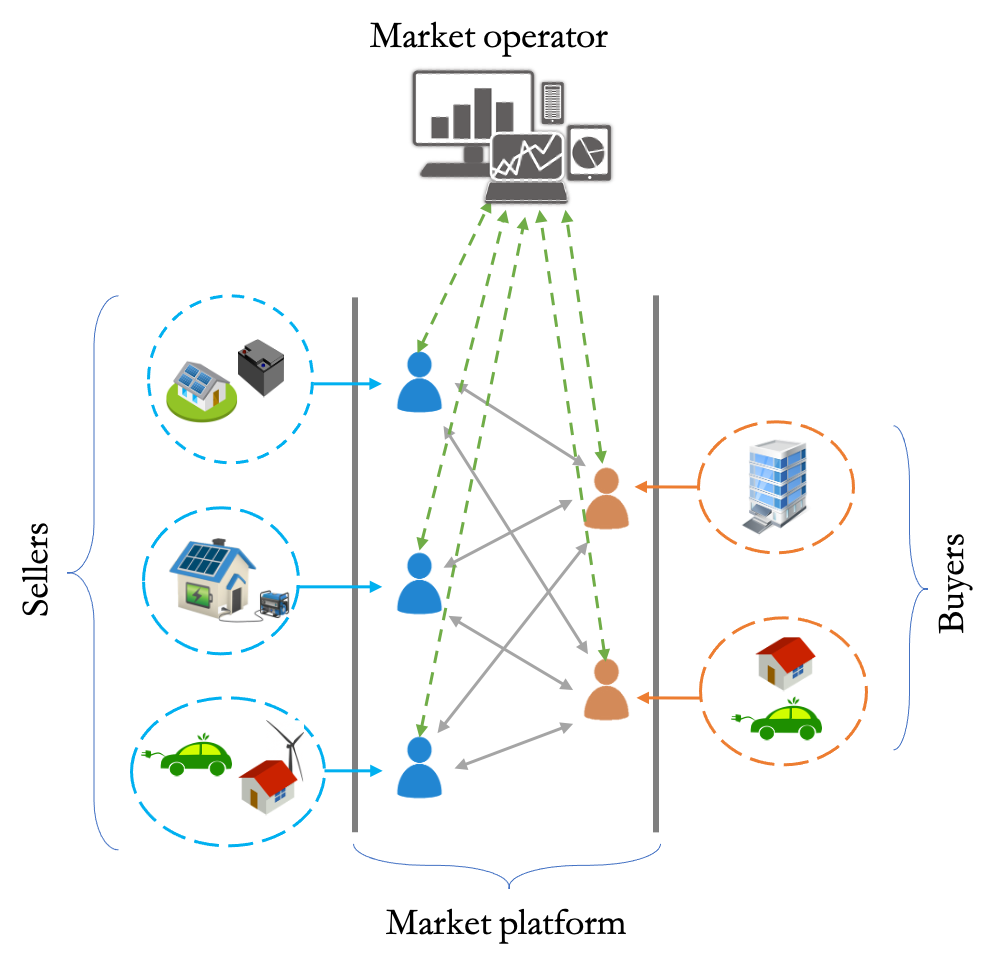}
\caption{ {Illustrative scheme of a bilateral P2P energy market.}}
\label{fig: market structure}
\vspace{-0.1mm}
\end{figure}
 \noindent Let $c_j$ denote the valuation of a seller $j \in \mathcal{I}_\mathcal{S}$ for each unit (e.g. 1 KWh) of energy and let $s_j$ represent the total energy offered; then, an offer of a seller $j$ is given by a pair $(c_j, s_j)$. Similarly, we denote the energy demand of a buyer $i \in \mathcal{I}_\mathcal{B}$ by $d_i$ and his valuation for the energy offered by seller $j$ by  $\alpha_{i,j}p_i$ where, $\alpha_{i,j}$ is the preference factor assigned by buyer $i$ to seller $j$. The preference factor allows a buyer to differentiate between the offered energy and can depend on several metrics, such as  source of energy (green vs. brown), location of the seller, user rating, etc. Furthermore, $p_i$ is a base price that the buyer $i$ is willing to pay for each unit of energy, hence they present their bid as $(\alpha_{i,j}p_i, d_i)$. Next, we impose some practical limitations on the valuations of buyers and sellers to make our P2P market setup economically appealing for the participants. \\
Prosumers generally sell energy to the grid via a retailer, thus an offer higher than the retailer's remuneration  makes it favorable for the sellers to join the P2P market instead. Furthermore, the rationality of the buyer demands that his offers are not higher than the cost of energy from the grid. Let $g_\text{b}$ and $g_\text{s}$ denote the buying price and the selling price of energy provided by the grid, respectively. Then, the buyer $i$ should offer a seller $j$ a higher energy price than that of the grid, but not more than the grid's selling price, i.e.,
\vspace{-0.1mm}
\begin{equation}\label{eq: condition 1}
 \alpha_{i,j}p_i \in (g_\text{b}, g_\text{s}]. 
 \vspace{-0.1mm}
\end{equation}

Analogously, we require the selling price of seller $j$ to fulfill the similar limitations,
\vspace{-0.1mm}
\begin{equation}\label{eq: condition 2}
c_j \in [g_\text{b}, g_\text{s}). 
\vspace{-0.1mm}
\end{equation}
We remark that similar assumptions are also made in \cite{tushar2019grid}.  {These assumptions are reasonable as the feed-in tariffs have seen a decreasing trend and were also discontinued in some regions \cite{tushar2018peer}. Nevertheless, if the feed-in tariff offered by the grid is very high, then it will impact the prosumer participation in the P2P market.}
\vspace{-0.1mm}
\subsection{Single-contract market}\label{subsec: single-contract}
In a single-contract P2P market setup, the buyers and sellers make one-to-one bilateral contracts that generate certain utility (value) for both. Let buyer $i \in \mathcal{I}_\mathcal{B}$ and seller $j \in \mathcal{I}_\mathcal{S}$ make a bilateral contract; then, the contract generates the value,
\vspace{-0.1mm}
\begin{equation}\label{eq: contract}
    v(i,j) = \max \{0, \alpha_{i,j}p_i - c_j\} \min\{s_j, d_i\}
    \vspace{-1.5mm}
\end{equation}

Let us elaborate on the formulation of the bilateral contract value in (\ref{eq: contract}). First, the contract is only viable when buyer's valuation of the energy is higher than seller's demand, i.e., $\alpha_{i,j}p_i > c_j$. If $s_j \geq d_i$, then the welfare generated by each traded unit is given by $\alpha_{i,j}p_i - c_j$ where, the total traded units are $d_i$. Now after the bilateral contract, the excess energy of the seller $(s_j - d_i)$ is sold to the grid. Analogously if $d_i > s_j$. We note that, the value of a non-viable contract will be zero. \\
Due to the bilateral structure of our P2P market, we can express the worth of possible contracts in a matrix form, which further allows us to model the market welfare maximization as an assignment problem.  Let  $M = [v(i,j)]_{(i,j) \in \mathcal{I}_\mathcal{B}\times \mathcal{I}_\mathcal{S}}$ be an assignment matrix where each element $v(i,j)$ represents the value of a bilateral contract between buyer $i$ and seller $j$. Then, we denote the corresponding assignment game by $\mathcal{M}=\left(\mathcal{I}_{B} \cup \mathcal{I}_{S}, v_{M}\right)$. The resulting value function of an assignment game $v_M(S)$ utilized by the market operator is given by the following assignment problem, for each $S \subseteq \mathcal{I}$:
\vspace{-0.1mm}
\begin{equation}\label{eq: assignment game}
\mathbb{P}(S):
\left\{ \quad
\begin{aligned}
 &\displaystyle \underset{\mu}{\max}  \sum_{i \in \mathcal{I}_{\mathcal{B}}\cap S} \sum_{j \in \mathcal{I}_{\mathcal{S} }\cap S} v(i,j) \mu_{i, j} \\
    & \displaystyle \:\: \mathrm{s.t.}  \sum_{i \in \mathcal{I}_{\mathcal{B}}\cap S}  \mu_{i, j} \leq 1 \qquad \quad \forall j \in \mathcal{I}_{\mathcal{S} }\cap S \\
    & \displaystyle \quad \;\;\:  \sum_{j \in \mathcal{I}_{\mathcal{S}}\cap S}  \mu_{i, j} \leq 1 \qquad \quad \forall i \in \mathcal{I}_{\mathcal{B} }\cap S 
\end{aligned}
\right.
\vspace{-0.1mm}
\end{equation}
with matching factors $\mu_{i, j} \in \{0, 1\}$, where $\mu_{i, j} = 1$ represents the matching between buyer $i$ and seller $j$.  {The problem (\ref{eq: assignment game}) determines the optimal assignment of buyers to sellers and the constraints imposed on the matching factors ensure that one buyer is matched to only one seller, i.e., one-to-one matching.} By using the results of the assignment problem in (\ref{eq: assignment game}) the market operator can evaluate the core of the game, as in (\ref{core}). We note that, even though the assignment problem in (\ref{eq: assignment game}) is a combinatorial optimization problem, because of its special structure  it can be solved in polynomial time using specifically designed algorithms like the Hungarian algorithm. Next, we list the notable features of our bilateral P2P market design:  
\begin{itemize}
    \item Existence: There always exist a set of bilateral contracts which is satisfactory for all of self-interested participants. In other words, the core of a bilateral P2P energy market is always non-empty (Remark \ref{lem: non empty core}).
    \item Product differentiation: Buyers can prioritise sellers or the categories of sellers via preference factors $\alpha_{i,j}$, based on the desired criteria  (e.g. green energy).
    \item Mechanism properties: Market formulation ensures the desirable economic properties of the clearing mechanism.
    \item Social optimality: The bilateral contracts maximize the overall welfare of the market and the contract price is negotiated internally between buyers and sellers.
\end{itemize}
\vspace{-0.1mm}
\subsection{Multi-contract market}\label{subsec: multi contract}
The formulation of a P2P market presented in Section \ref{subsec: single-contract}  is a one step single-contract bilateral market where each buyer can make an energy trade with only one seller and vice versa. Therefore, even though the proposed formulation maximizes the overall welfare, the market participants on both sides can have partially fulfilled energy trades. Hence, in this section we extend the single-contract P2P energy market to accommodate multiple contracts between buyers and sellers which in turn allows for the complete fulfilment of energy trades.\\
For an assignment market, we model multiple contracts between buyers and sellers by granulation of energy demand or offered into the units (packets) of fixed size (e.g. 1 KWh). Consequently, the matching takes place between these units of energy. Another way of looking at this setup is that each market participant (buyer or seller) is represented in the market by multiple agents, with each agent offering or demanding single unit of energy. Hence, the number of agents representing each participant in the market are equal to the number of energy units offered/demanded. To provide further flexibility, in our multi-contract model, we allow participants to associate different trading characteristics (e.g. valuation, energy source) to each traded unit. For example, a seller can offer energy units from RES (green) and an energy storage (possibly brown); similarly a buyer can bid higher for the energy needed for the critical tasks and lower for the deferrable tasks.  In the mathematical formulation, we interpret each agent as an independent seller or buyer, thus the resulting value $v(i,j)$ generated by contracts is similar to the expression in (\ref{eq: contract}) for single traded unit ($d_i = s_i = 1$ unit), i.e., $\textstyle  v(i,j) = \max \{0, \alpha_{i,j}p_i - c_j\},$
and the corresponding assignment game $\mathcal{M}=\left(\mathcal{I}_{B} \cup \mathcal{I}_{S}, v_{M}\right)$ is solved using the assignment problem in (\ref{eq: assignment game}). \\
In the multi-contract setup, in addition to maximizing the overall welfare, we can maximize the energy traded inside the bilateral P2P energy market by varying the level of granulation. Specifically, if all the contracts are viable and  $ \textstyle \sum_{i \in \mathcal{I}_{\mathcal{B}}} d_i >  \sum_{j \in \mathcal{I}_{\mathcal{S}}} s_j $, then by selecting the appropriate size of single energy unit, we can ensure that the total energy offered will be traded inside the P2P market, bilaterally.\\ 
The additional features of this multi-contract setup, however come at the cost of higher computational burden due to increased number of agents, representing the trade of each energy unit. We note that the appropriate size of the energy unit, decided by the market operator, can limit the number of agents and the associated computational burden. Furthermore, single-contract and multi-contract setup can be deployed in different contexts. For example, the former is more suitable for implementation in the larger scales, whereas, the later can bring additional features for localized implementation with lower number of participants such as in energy communities. 
\vspace{-0.1mm}
\section{Distributed solution mechanism}\label{sec: solution mechanism}
After the market operator solves (\ref{eq: assignment game}), in the second stage of our design, the participants negotiate among themselves for a bilateral agreement on the trading price. Here, our goal is to enable the participants to autonomously reach a consensus on a set of bilateral contract prices such that no party can raise any objection on the contracts. Therefore, we propose a novel negotiation mechanism that allows for faster convergence rates, thus it is suitable for both single and multi-contract market setups. We model our algorithm in a distributed architecture, where a central market operator with complete information of the game initially transmits information of the bounding sets in (\ref{eq: bounding set}) to the respective agents (market participants).  {The knowledge of a bounding set implies that each agent knows the values of their own coalitions only.} After receiving the required information, each agent distributedly proposes a payoff allocation for all the agents. We prove that even with the partial information available, the proposed solution mechanism converges to a stable payoff distribution. The mechanism for the proposed bilateral P2P electricity market is detailed in Figure \ref{flow chart}. In particular, we design a distributed fixed-point algorithm, using which the agents can reach consensus (\ref{eq: consensus}) on a payoff distribution in the core of the P2P market in (\ref{eq: assignment game}). 

\tikzstyle{decision} = [diamond, draw, aspect=3.5,
    text width=5em, text badly centered, node distance=2cm, inner sep=0pt]
\tikzstyle{block} = [rectangle, draw,
    text width=5.2em, text centered, minimum height=1em]
\tikzstyle{blockthin} = [rectangle, draw,
    text width=4em, text centered, minimum height=1em]
\tikzstyle{blocklong} = [rectangle, draw,
    text width=5cm, text centered, minimum height=1em]
\tikzstyle{blocklong2} = [rectangle, draw,
    text width=7cm, text centered, minimum height=1em]   
\tikzstyle{line} = [draw, -latex']

\begin{figure}
    \centering
    
\begin{tikzpicture}[node distance = 1.15cm, auto]
    \node [blocklong2] (init) {\small market operator evaluates $v(i,j)$ by (\ref{eq: assignment game}) at $k=0$};

    \node [block, below of=init] (agent2) {\small agent 2, $\boldsymbol{x}_2^0$};
    \node [block, right of=agent2, node distance=3cm] (agentN) {\small agent $N$, $\boldsymbol{x}_N^0$};
    \node [block, left of=agent2, node distance=3cm] (agent1) {\small agent 1, $\boldsymbol{x}_1^0$};
    
    \node [blocklong, below of=agent2] (share) {\small share $\boldsymbol{x}_i^k$ with neighbour agents} ;
    \node [blocklong, below of=share] (evaluate) {\small agents update proposals $\boldsymbol{x}_i^k$ by (\ref{main_it})} ;
    
    \node [decision, below of=evaluate, node distance=1.15cm] (decide) {\small $\boldsymbol{x}^k \in \mathcal{X}^*$?};
    \node [blockthin, right of=decide, node distance=3cm] (update) {\small $k \leftarrow k+1$};
    \node [blockthin, left of=decide, node distance=3cm] (stop) {\small stop};
    \node at ($(agent2)!.5!(agentN)$) {\ldots};
    
    \path [line]  (init)  -- node [left] {\small $\mathcal{X}_1 \;\; $} (agent1);
    \path [line] (init) -- node  {\small $\mathcal{X}_2$} (agent2);
    \path [line] (init) -- node [near end] {\small $\mathcal{X}_N$} (agentN);
    
    \path [line] (agent2) -- (share);
    \path [line] (agent1) -- (share);
    \path [line] (agentN) -- (share);
    \path [line] (share) -- (evaluate);
    \path [line] (evaluate) -- (decide);
    \path [line] (decide) -- node {no} (update);
    \path [line] (update) |- (share);
    \path [line] (decide) -- node {yes}(stop);
\end{tikzpicture}

    \caption{ {Flowchart of our proposed bilateral P2P market mechanism.}}
    \label{flow chart}
\vspace{-0.1mm}
\end{figure}
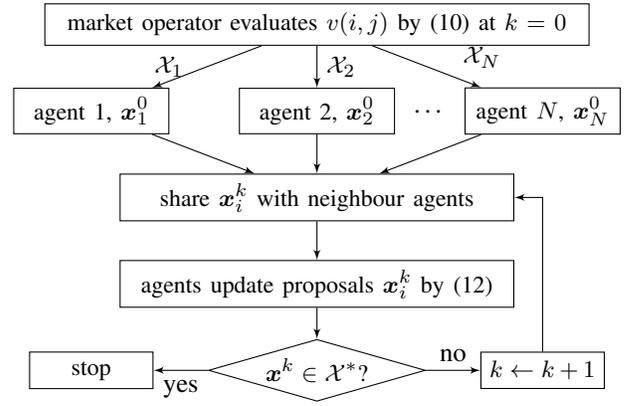
\vspace{-0.1mm}
\subsection{Distributed negotiation mechanism}\label{subsec: Solution mechanism}
We consider a bilateral negotiation process in which, at each negotiation step $k$, a buyer (seller) can communicate with a set of \textit{neighbouring} sellers (buyers) to bargain for their payoff. Therefore,  
we model their communication over a time-varying network represented by a bipartite graph $G^{k}=(\mathcal{I}_\mathcal{B} \times \mathcal{I}_\mathcal{S}, \mathcal{E}^{k}) $, where for $i \in \mathcal{I}_\mathcal{B}$ and $j \in  \mathcal{I}_\mathcal{S}$,  $(i, j) \in \mathcal{E}^{k}$ means that there is an active link between buyer $i$ and seller $j$ at iteration $k$ and they are then referred as neighbours. 
We assume that at each iteration $k$ an agent $i$ observes only the proposals of its neighbouring agents. Furthermore, we assume that each buyer-seller pair communicates at least once during a time period of length $Q$ (arbitrarily large), which ensures that the agents communicate sufficiently often. In other words, we assume that the union of the communication graphs over a time period of length $Q$ is connected. This assumption is fairly common in multi-agent coordination, e.g. \cite[Assumption 3.2]{nedic2017achieving}. 
\begin{assum}[$Q-$connected graph]\label{asm: Q-con }
For all $k \in \mathbb{N}$, the union graph $(\mathcal{I}_\mathcal{B} \times \mathcal{I}_\mathcal{S}, \cup_{l=1}^{Q} \mathcal{E}^{l+k})$  is strongly connected for some integer $Q \geq 1$.  $\hfill \square$
\end{assum}
The edges (links) in the communication graph $G^{k}$ are weighted using an adjacency matrix $W^{k} = [w_{i,j}^k]$, whose element $w_{i,j}^k$ represents the weight assigned by agent $i$ to the payoff proposal of agent $j$, ${ \boldsymbol{x}}_j^{k}$ where, for some $j$, $w_{i,j}^k = 0$ implies that the agent $i$ does not negotiate with agent $j$ at iteration $k$, i.e., $(i, j) \notin \mathcal{E}^{k} $. We note that in a P2P market the buyers (sellers) do not negotiate among themselves  hence, $w_{i,j}^k = 0, \text{ for all } (i, j) \in \mathcal{I}_\mathcal{B}$    $(\mathcal{I}_\mathcal{S}) $.  Furthermore, to ensure that all the agents have sufficient influence on the resulting payoff distribution, we assume the adjacency matrix to be doubly stochastic with positive diagonal, which means that an agent always gives some weight to his previous proposal.
\begin{assum}[Stochastic adjacency matrix]\label{asm: graph}
 For all $k \geq 0$, the adjacency matrix $W^{k} = [w_{i,j}^k]$ of the communication graph $G^k$ is doubly stochastic, i.e., \(\sum_{j=1}^{N} w_{i,j}=\sum_{i=1}^{N} w_{i,j}=1\), its diagonal elements are strictly positive, i.e., $w_{i,i}^k > 0, \text{ for all } i \in \mathcal{I}$ and $\exists$ $\gamma > 0$ such that $w_{i,j}^k \geq \gamma$ whenever $w_{i,j}^k > 0$ \cite[Assumption 3.3]{nedic2017achieving}. $\hfill \square$

\end{assum}
We further assume that the elements of the communication matrix $W^{k}$ take values from a finite set hence, finitely many adjacency matrices are available.
\begin{assum}[Finitely many adjacency matrices]\label{asm: fixed graph}
The adjacency matrices $\{W^k\}_{k \in \mathbb{N}}$ of the communication graphs belong to $\mathcal{W}$, a finite family of matrices that satisfy Assumption \ref{asm: graph}, i.e., $W^k \in \mathcal{W}$ for all $k \in \mathbb{N}$.  $\hfill \square$
\end{assum}
This assumption on the adjacency matrices is purely technical and allows us to exploit important results from the literature, for proving convergence of our negotiation mechanism. We remark that the set of adjacency matrices can be arbitrarily large hence Assumption \ref{asm: fixed graph} poses no practical limitation on our negotiation mechanism, which we propose next. \\
At each negotiation step $k$, an agent $i$ bargains by proposing a payoff distribution $\boldsymbol{x}_i^k \in \mathbb{R}^N$, for all the agents. To evaluate a proposal, they first take an average of the estimates of neighboring agents, $\boldsymbol{x}_j^k$ such that $(i,j) \in \mathcal{E}^{k}$, weighted by an adjacency matrix $W^k$, $\sum_{j=1}^{N} w_{i,j}^k \boldsymbol{x}_{j}^{k}$.  Next, agent $i$ utilizes a partial game information in the form of a bounding half-space $H_i^k \in \mathcal{H}_i$ as in (\ref{eq: hyperplanes}) of a bounding set in (\ref{eq: bounding set}). An agent selects the half-spaces from the set $\mathcal{H}_i$ such that each bounding half-space appears at-least once in every $Q$ negotiation steps with $Q$ as in Assumption \ref{asm: Q-con }. In practice, one way of selecting these half-spaces can be a predefined sequence that is arbitrarily chosen by each agent. 
Finally, agent $i$ projects the average $\hat{\boldsymbol{x}}_i^{k}:= \sum_{j=1}^{N} w_{i,j}^k \boldsymbol{x}_{j}^{k} $ on the bounding half-space. Thus, the algorithm reads as
\vspace{-0.1mm}
\begin{equation}\label{proj_main_it}
  \boldsymbol{x}_i^{k+1}= \mathrm{proj}_{H_i^k}\textstyle (\hat{\boldsymbol{x}}_i^{k}). 
\end{equation}
 {The protocol in (\ref{proj_main_it}) allows agents to propose a payoff at each negotiation step that is acceptable for them.} Let us further generalize the iteration in (\ref{proj_main_it}) by replacing the projection operator, $\mathrm{proj}(\cdot)$, with a special class of operators namely, paracontractions. This generalization enables the agents to choose any paracontraction operator $T_i^k$, for evaluating a payoff proposal $\boldsymbol{x}_i^k$, which in turn allows for a faster convergence to the interior of the core in (\ref{core assignment}). The latter is an important feature because the interior of the core is associated with the fairness of the payoff in assignment games. Specifically, for each $i \in \mathcal{I}$, we propose the negotiation protocol  $\boldsymbol{x}_i^{k+1}= T_i^{k}\textstyle (\hat{\boldsymbol{x}}_i^{k})$,
that in collective form, reads as the fixed-point iteration
 \begin{equation}\label{main_it}
    { \boldsymbol{x}}^{k+1} = \boldsymbol{T}^k (\boldsymbol{W}^k { \boldsymbol{x}}^{k}),
\end{equation}
 where  $\boldsymbol{T}^k (\boldsymbol{x}):= \mathrm{col}(T_1^{k}(\boldsymbol{x}_1), \ldots,  T_N^{k}(\boldsymbol{x}_N))  $ and $\boldsymbol{W}^k := W^{k} \otimes I_N $ represents an adjacency matrix. In (\ref{main_it}), we require the paracontraction operator $T_i^k$ to have $H_i^k \in \mathcal{H}_i$ in (\ref{eq: hyperplanes}) as fixed-point set, i.e., $ \mathrm{fix}(T_i^k) = H_i^{k} $. 
 \begin{assum}[Paracontractions]\label{asm: fixed points of M}
 For $k \in \mathbb{N}$, $\boldsymbol{T}^k$ in (\ref{main_it}) is such that $T_i^k \in \mathcal{T}$, where $\mathcal{T}$ is a finite family of paracontraction operators such that $ \mathrm{fix}(T_i) = H_i$ with $H_i \in \mathcal{H}_i$ in (\ref{eq: hyperplanes}). $\hfill \square$
 \end{assum}
Here, for utilizing the negotiation mechanism in iteration (\ref{main_it}), an agent can choose any operator $T_i$ that satisfies Assumption \ref{asm: fixed points of M}. This choice can affect the speed of convergence, as demonstrated in Section \ref{sec: Numerical simulations}, and also the specific limit point inside the core. Examples of paracontractions include the projection on a closed convex set $C$, $\mathrm{proj}_C(\cdot)$, and the convex combination of projection and over-projection operators, i.e., $T = (1 - \beta) \mathrm{proj}_{C}(\cdot) + \beta \mathrm{overproj}_{C}(\cdot)$ with $\beta \in [0, 1)$.\\
We also assume that each $\boldsymbol{T}^k \in \mathcal{T}^N$ appears at least once in every $Q$ iterations of (\ref{main_it}), with $Q$ as in Assumption \ref{asm: Q-con }.
\begin{assum}\label{asm: Q admissible bargaining} 
Let $Q$ be the integer in Assumption \ref{asm: Q-con }. The operators $(\boldsymbol{T}^k)_{k \in \mathbb{N}}$ in (\ref{main_it}) are such that, for all $n \in \mathbb{N}$, $\bigcup_{k=n}^{n+Q}\{\boldsymbol{T}^k\} = \mathcal{T}^N$, with $\mathcal{T}$ as in Assumption \ref{asm: fixed points of M}. $\hfill \square$
\end{assum}
Next, we formalize our main convergence result for the negotiation mechanism in (\ref{main_it}).
\begin{theorem}[Convergence of negotiation mechanism]\label{theorem: main}
Let Assumptions \ref{asm: Q-con }$-$\ref{asm: Q admissible bargaining} hold. Let $\mathcal{X}^*:=\mathcal{A}\cap \mathcal{C}_M^N$ with $\mathcal{A}$ as in (\ref{eq: consensus}) and $\mathcal{C}_M$ being the core in (\ref{core assignment}). Then, {starting from any $\boldsymbol{x}^0 \in \mathbb{R}^{N^2}$,} the sequence \((\boldsymbol{x}^{k})_{k=0}^{\infty}\) generated by the iteration in (\ref{main_it}) converges to {some} $\bar{\boldsymbol{x}} \in \mathcal{X}^*$. $\hfill \square$
\end{theorem}
We provide the proof of Theorem \ref{theorem: main} in Appendix. We remark that, presenting the mechanism as a fixed-point iteration and in terms of operators allows us to utilize results from operator theory to keep our convergence analysis general and brief.
\vspace{-0.1mm}
\subsection{Technical discussion}\label{subsec: technical discussion}
Theorem \ref{theorem: main} shows that the repeated proposals by all agents, generated by our negotiation mechanism, eventually reach an agreement on a payoff that belongs to the intersection of the bounding sets, i.e., the core. This core payoff allows us to compute the \textit{stable} contract prices. Let the payoff of buyer $i$ and seller $j$ be $x_i$ and $x_j$ respectively then, the contract price is $\lambda_{i,j} = \alpha_{i,j}p_i - x_i$. We note that the core of an assignment game has a special structure with two extreme points, i.e. buyer optimal and seller optimal at its boundary. A buyer optimal payoff is the worst core payoff for the seller side and vice versa. Thus, the payoff in the interior of the core corresponds to a fairer allocation and a consensus on such allocation can be achieved via the proposed algorithm in (\ref{main_it}).  \\
Let us now mention the features of our algorithm that enhance its practicality.{First, for the negotiation, the market participants do not require full information of the game but only the values of their own contracts represented by the bounding sets in (\ref{eq: bounding set}), which is privacy preserving. Such a lower information requirement of our mechanism is a considerable benefit over the algorithm presented in \cite{bauso2015distributed}, which requires each participant to have complete information of the corresponding core set in (\ref{core}).} Secondly, utilizing the half-spaces $H_i \in \mathcal{H}_i$ as the fixed-point sets of the operators $T_i \in \mathcal{T}$ in (\ref{main_it}) allows us to design $\mathcal{T}$ as a set of linear operators. For example, let $\mathbf{e}_k$ be the vector of coefficients of the inequality that defines the bounding half-space, i.e., $H_{i}^{k}=\{y \in \mathbb{R}^{N} \mid \mathbf{e}_k^{\top} y \geq \eta\}.$ Then,  we can write the iteration in (\ref{proj_main_it}) as 
\vspace{-0.1mm}
$$ \displaystyle \boldsymbol{x}_i^{k+1}=  \hat{\boldsymbol{x}}_i^{k} +\frac{\eta- \mathbf{e}_k^{\top} \hat{\boldsymbol{x}}_i^{k} }{\|\mathbf{e}_k\|^{2}} \mathbf{e}_k, \vspace{-2mm} $$
if $\hat{\boldsymbol{x}}_i^{k} \notin H_{i}^{k}$ \cite[Example 28.16]{bauschke2011convex}.
This closed-form expression reduces the computational burden of our algorithm greatly since no optimization problem should be solved at each iteration. Clearly, both privacy preservation and low computational burden are highly desirable features of a market mechanism. 
\vspace{-0.1mm}
\section{Numerical simulations}\label{sec: Numerical simulations}
In this section, we simulate the proposed bilateral P2P energy market with prosumers employing the negotiation mechanism designed in Section \ref{sec: solution mechanism}. We conduct the analysis for the time slots that incur peak prices and have considerable PV generation and compare it with the conventional approach of trading with the grid via aggregators and retailers, to demonstrate the effectiveness of our algorithm and show the economic benefits for the prosumers. 
Next, the size of the time slots should be decided by the market operator considering the variation in the demand and production of energy at prosumer level. For our economic analysis, we use hourly time slots of four peak hours for each day over a week and for convergence analysis we consider 100 scenarios of prosumer demand and generation to report their average and spread of samples. Furthermore, as our focus is on the economic and algorithmic design of the market mechanism, we do not consider network constraints (as also done in \cite{tushar2018peer}, \cite{tushar2019motivational}) and remark that their incorporation would not effect the resulting market properties. 
\vspace{-0.1mm}
\begin{table}[]
 \centering
\caption{ {Profiles of buyers and sellers}}
\label{tab: profiles}
\begin{tabular}{ccccl} \toprule
 &  \multicolumn{2}{c}{Concern for} & & \\
\textbf{Buyers} & environment & rating & \textbf{Sellers} & Energy source \\ \midrule
B1& \checkmark & $\times$ & S1 & PV (green) \\ \midrule
B2& $\times$ & $\times$ & S2 & Storage (brown) \\ \midrule
B3& $\times$ & \checkmark & S3 & PV (green)  \\ \midrule
B4& \checkmark & \checkmark & S4 & fossil (brown)\\ \bottomrule
 \end{tabular}
\vspace{-0.1mm}
 \end{table}
\subsection{Simulation setup:} 
We consider 4 residential prosumers with energy deficiency and 4 with surplus to act as buyers and sellers, respectively for each time slot. During different sessions of the market, prosumers can vary between the roles of sellers and buyers, depending on their energy profiles.  However, during each session (time slot) a prosumer acts as either a buyer or a seller. The energy deficiency and surplus of each prosumer lies within the range of [2, 8]. Next, we purposefully build the profiles of prosumers to show diverse participation and to emphasize various features of our P2P market designs. To buy energy for a given time-slot, a buyer $i$ enters a P2P market with its bid and demand $(\alpha_{i,j}p_i, d_i)$ as in (\ref{eq: contract}) where, the factor $\alpha_{i,j}$ represents his preference valuation for the energy offered by seller $j$. For the design of preference factor, we let buyer specify his level of environmental concern (preference to the green energy) $\alpha^\text{g}_{i}$ on the scale of $\{0, \cdots, 5\}$  and his concern to seller's user rating $\alpha^\text{r}_j \in \{0, \cdots, 5\}$ by $\gamma^\text{r}_{i} \in \{0,1\}$, with 0 being indifference to the associated factor. Let us indicate the energy type of seller $j$ by $\gamma^\text{g}_{j} \in \{0,1\}$ with $1$ specifying green energy then, the preference factor is evaluated as $\alpha_{i,j} = 1 + 0.1 (\alpha^\text{g}_{i}\gamma^\text{g}_{j} +  \alpha^\text{r}_{j} \gamma^\text{r}_{i})$. The value of $\alpha^\text{g}_{i}$ is randomly chosen for the buyers who include the environmental concern in their profiles, given in Table \ref{tab: profiles}, and the consumer rating of each seller is chosen randomly from the range of $[3,5]$. We note that our design of preference factor is arbitrary and the market operator can design it differently to include other considerations.\\
The buyers choose base valuation of the energy $p_i$ such that their bid is higher than the grid's buying price $g_\text{b} = 0.05 $ £/kWh and not more than the grid's selling price $g_\text{s} = 0.17 $ £/kWh as in condition (\ref{eq: condition 1}) \cite{han2018incentivizing}. Furthermore, the sellers choose their valuation $c_j$ less than the grid's selling price $g_\text{s}$ as in (\ref{eq: condition 2}).
\vspace{-0.1mm}
\subsection{Bilateral P2P energy market}
In our P2P market setup, at the first stage, the market operator performs the optimal matching as formulated in (\ref{eq: assignment game}), which results in the optimal buyer-seller pairs. Next, the participants adopt the negotiation mechanism in (\ref{main_it}) to mutually decide the bilateral contract prices. In Figure \ref{fig: convergence}, we show, for a particular time slot, the convergence of our negotiation algorithm using the operators $\mathrm{proj}_{H^k}(\cdot)$ and $T_{H^k}:= (1 - \beta) \mathrm{proj}_{H^k}(\cdot) + \beta \mathrm{overproj}_{H^k}(\cdot)$, for single-contract and multi-contract market setups. We report the average of 100 samples of convergence trajectories obtained by varying energy offer and demand conditions and the sequence of half-spaces, i.e., negotiation strategy of each agent (see Section \ref{subsec: Solution mechanism}).  
We can observe that the operator $T_{H^k}$ results in the faster convergence, as we claimed in Section \ref{sec: solution mechanism}. We remark that the convergence speed of negotiation in a multi-contract market decreases with an increase in the level of energy granulation.  {In Figure \ref{fig: convergence_shade} we plot the spread of the sample trajectories  to illustrate the best and worst convergence scenario for both single-contract and multi-contract setups using the operator $T_{H^k}$. Next, we benchmark the computational performance of our algorithm. We choose a static case of a distributed bargaining algorithm proposed in \cite{nedic2013} for payoff allocation in coalitional games, as a benchmark. In Figure \ref{fig: benchmark_comparison}, we present the trajectories of our algorithm and the benchmark. Since, the benchmark algorithm utilizes the whole bounding set instead of just the bounding half space in (\ref{eq: hyperplanes}) at each iteration, it proceeds faster initially. However, in the long run, our proposed algorithm performs better. Lower information requirement in our approach makes the execution of a negotiation step considerably faster. In this simulation scenario, the average execution time of a negotiation step is about $40\times$ times faster than the benchmark.}

\begin{figure}[t]
\centering
\includegraphics[width=0.9\linewidth]{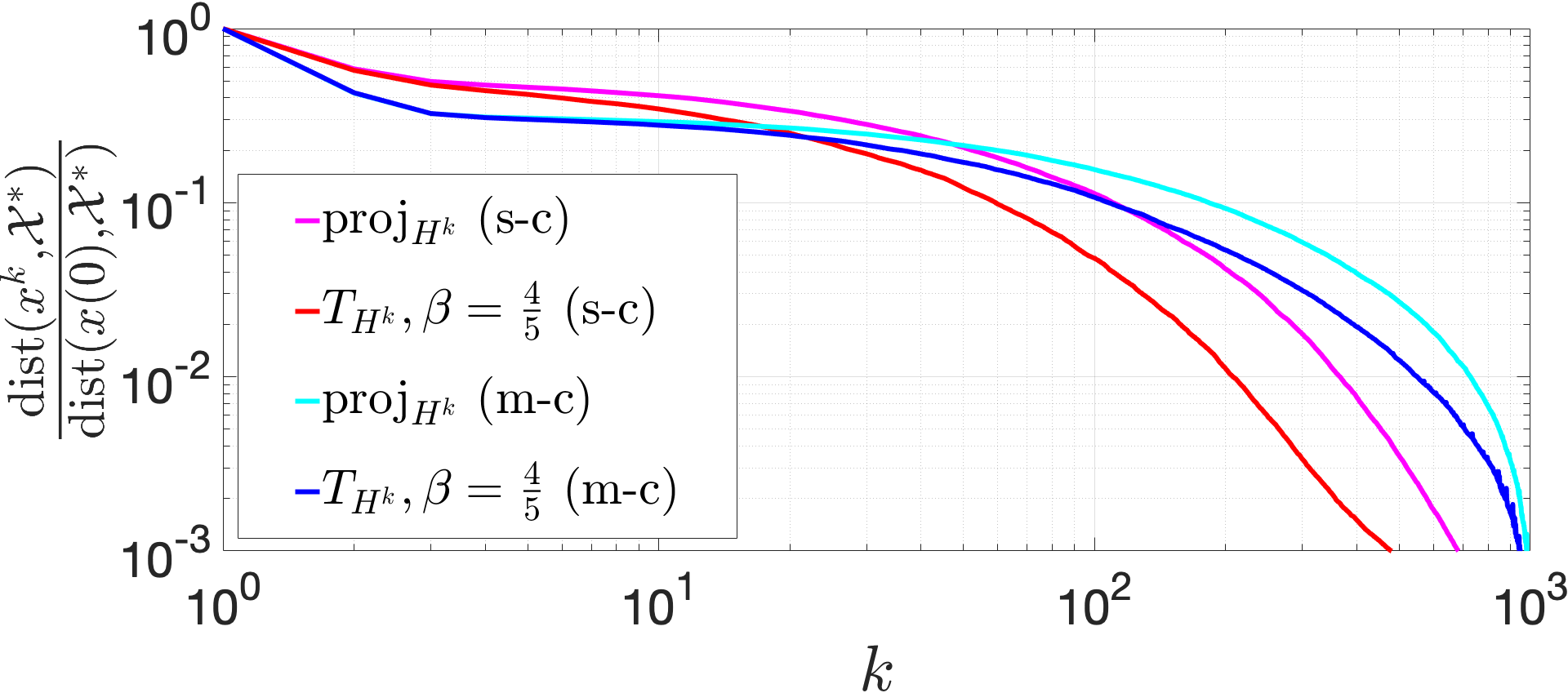}
\caption{\footnotesize Trajectories of $\mathrm{dist}(\boldsymbol{x}^k, \mathcal{X}^*)/\mathrm{dist}(\boldsymbol{x}^0, \mathcal{X}^*)$ for distributed bilateral negotiation with operators $\mathrm{proj}_{H^k}$ and $T_{H^k}:= (1 - \beta) \mathrm{proj}_{H^k}(\cdot) + \beta \mathrm{overproj}_{H^k}(\cdot)$ for single-contract (s-c) and multi-contract (m-c) market setups. }
\label{fig: convergence}
\vspace{-0.1mm}
\end{figure}

\begin{figure}[t]
\centering
\includegraphics[width=0.9\linewidth]{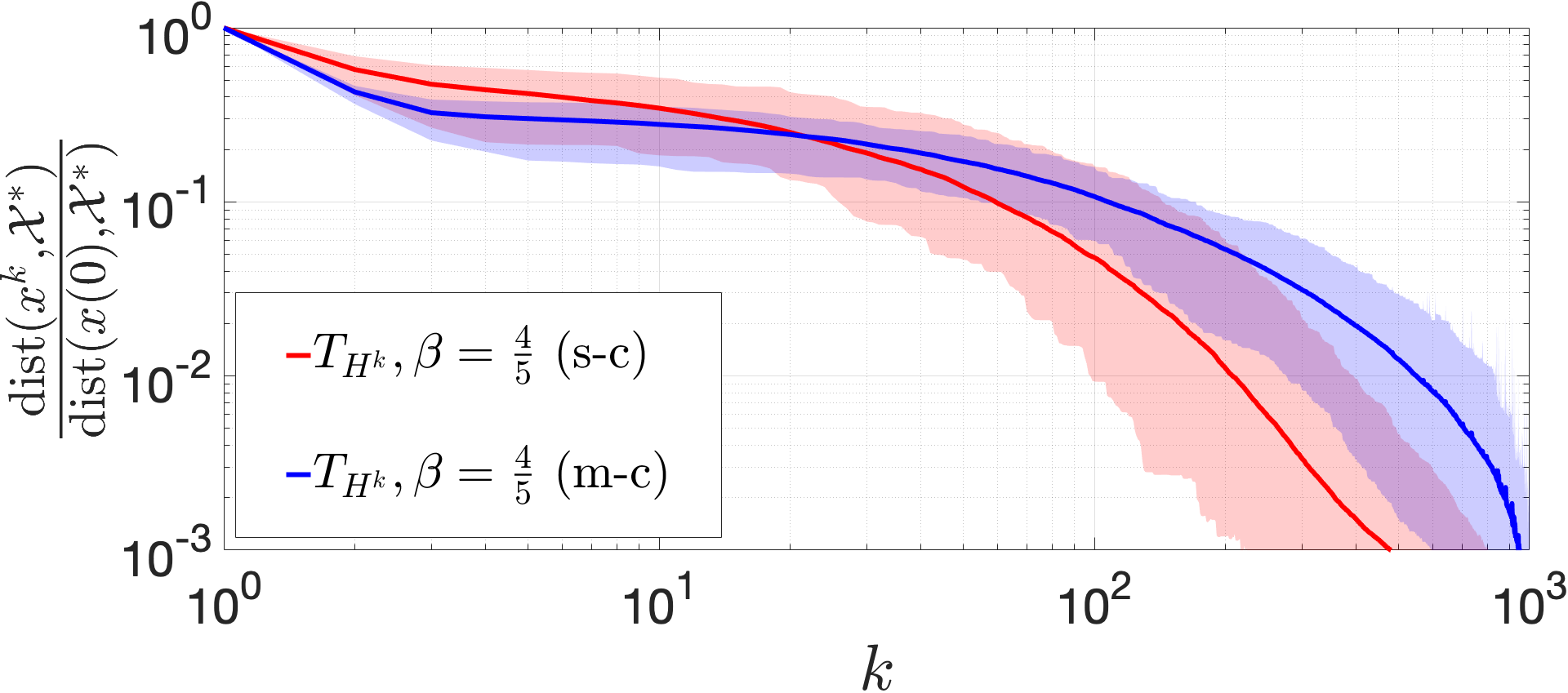}
\caption{ {Sampled average of the trajectories in Figure \ref{fig: convergence}, with spread of samples shown by shaded region.}}
\label{fig: convergence_shade}
\vspace{-0.1mm}
\end{figure}

\begin{figure}[t]
\centering
\includegraphics[width=0.9\linewidth]{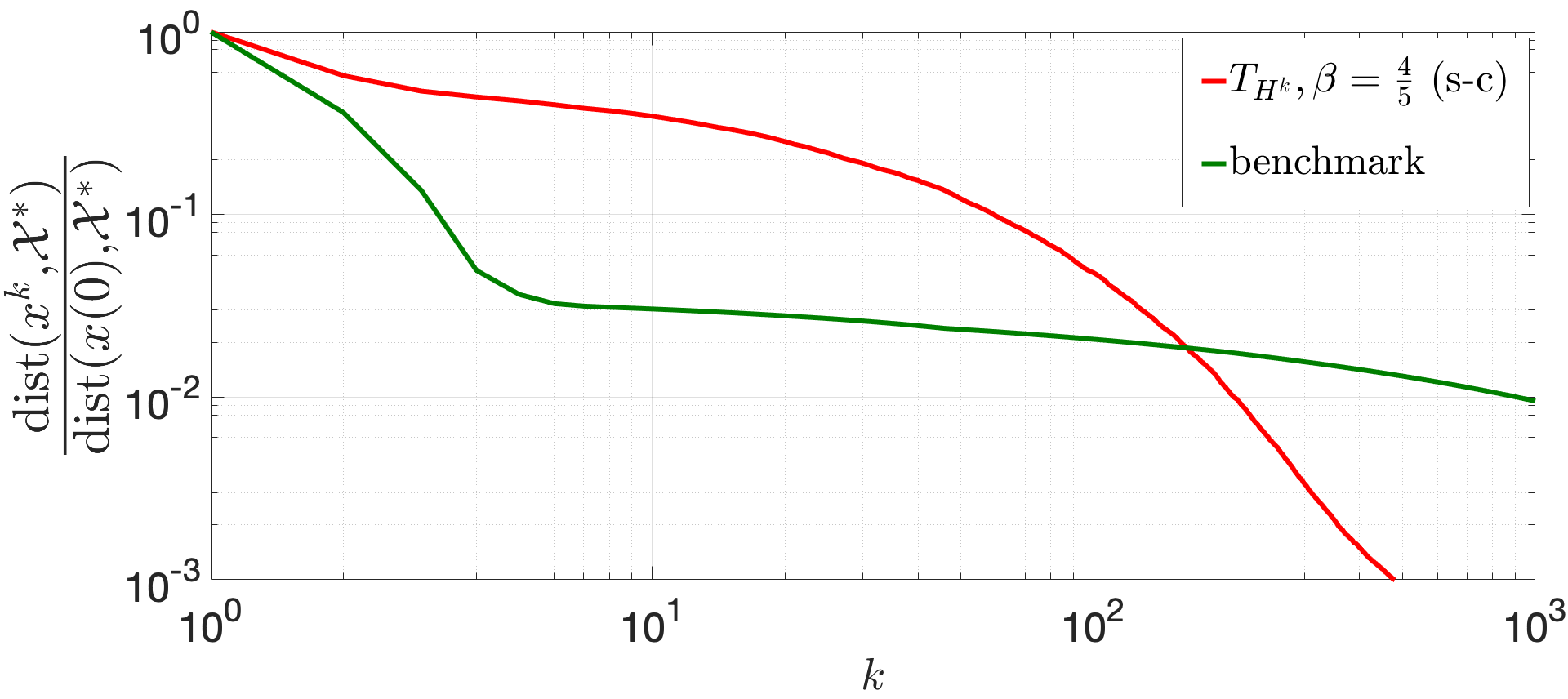}
\caption{ {Comparison of the trajectories of $\mathrm{dist}(\boldsymbol{x}^k, \mathcal{X}^*)/\mathrm{dist}(\boldsymbol{x}^0, \mathcal{X}^*)$ for our distributed bilateral negotiation (red) and the benchmark (green).}}
\label{fig: benchmark_comparison}
\vspace{-0.1mm}
\end{figure}

\begin{figure}[t]
\centering
\includegraphics[width=0.9\linewidth]{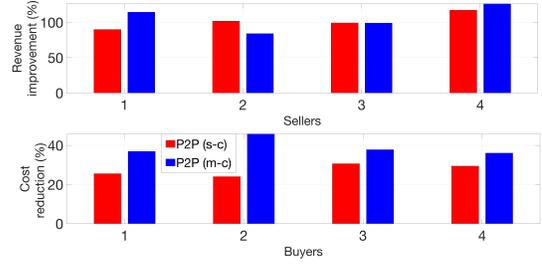}
\caption{Average revenue improvement (sellers) and average cost reduction (buyers) in single-contract (s-c) and multi-contract (m-c) P2P setups compared to trading with the grid.}
\label{fig: revenue}
\vspace{-0.1mm}
\end{figure}

\noindent Now, to evaluate the economic benefit for prosumers we set up P2P markets for peak hours and show the average change in revenues and costs of sellers and buyers, respectively, in Figure \ref{fig: revenue}. The revenues of sellers are higher and the costs of buyers are lower in both the market setups compared to the trade with the grid/retailer. We also observe that for the market with a similar category of participants (e.g. residential) and adequate participation from both sides (buyers and sellers), the economic performance of single-contract and multi-contract markets is comparable. This is due to the fact that similar excesses and demands of energy reduce the trade outside the P2P market. We also remark that moving from a single-contract to a multi-contract market can maximize the total amount of energy traded inside the P2P market hence increasing the overall market welfare but it does not guarantee individual improvements for all parties. This is because a random payoff inside the core set in (\ref{core assignment}) can assign a higher share of the value generated by a buyer-seller pair to either side of the market. For instance, this can be observed in Figure \ref{fig: revenue} for seller 2.  In Figure \ref{fig: trading}, we observe that the proposed P2P market designs strongly encourage prosumer participation in bilateral energy trading and, in particular, the multi-contract setup increases the internal energy trade.\\
 {Now, observe from (\ref{eq: assignment game}) that allowing for product differentiation (e.g. on an environmental or social basis) can increase the overall social welfare of the market compared to mere economic considerations. This increase comes from higher user satisfaction which is achieved by catering to their personal preferences. Therefore, we define a metric of user satisfaction as the number of times the buyers with green energy preferences are matched to the green sellers in 100 scenarios of our study. In Figure \ref{fig: user_satisfaction}, we compare this user satisfaction with a traditional market that only considers economic factors. The figure shows that product differentiation offers higher user satisfaction and thus encourages prosumer participation.}\\
 {Finally, in Table \ref{tab: market size}, we present the computational times of an agent's negotiation process in the proposed P2P market mechanism to numerically show the scalability with respect to the market size. We note that because of the distributed implementation, the negotiation protocol of agents run in parallel on their personal computational resource. Here, the simulations are executed in MATLAB 2020b installed on a laptop computer with 2.3 GHz Intel Core i5 and 8GB RAM. These numerical results can help in deciding how far ahead in time from the actual energy delivery should such markets operate, depending on the expected level of prosumer participation (e.g. the size of the energy community).}  {Furthermore, we can conclude that the adoption of the assignment game formulation provides an opportunity for the practical implementation of reasonably large P2P markets while guaranteeing contract prices that represent a competitive market equilibrium. We note that the regulator can also decide the market size systematically, e.g. on a geographical basis, and also put the eligibility criteria on the power capacity of the installed generation source. Such regulatory restrictions are often imposed on trading mechanisms for example, in Queensland, Australia, a prosumer cannot participate in a feed-in-tariff program if they have solar panels beyond 5kW capacity \cite{tushar2020coalition}}.

\begin{figure}[t]
\centering
\includegraphics[width=0.9\linewidth]{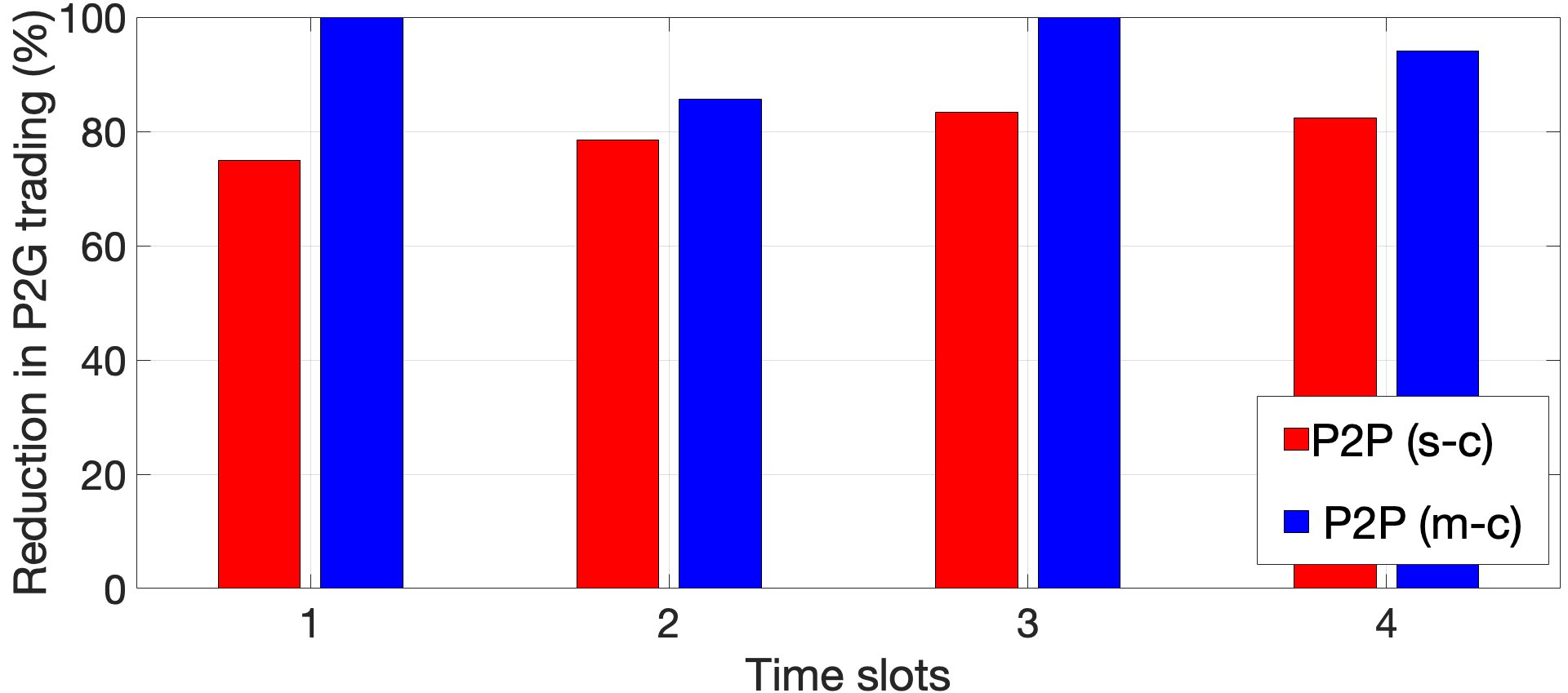}
\caption{Reduction in energy traded with the grid via single-contract (s-c) and multi-contract (m-c) P2P market setups.}
\label{fig: trading}
\end{figure}

\begin{figure}[t]
\centering
\includegraphics[width=0.9\linewidth]{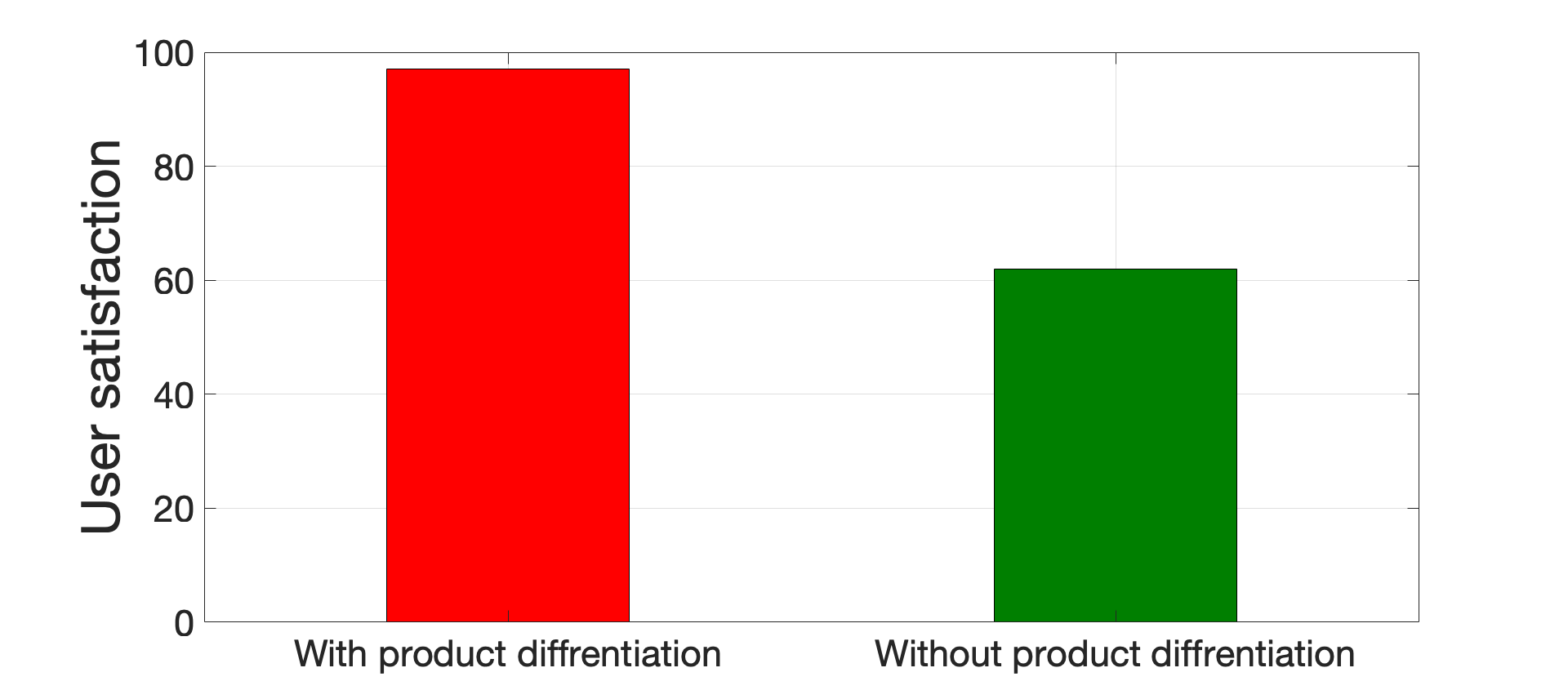}
\caption{ {User satisfaction for the buyers with a preference for green energy in P2P market setups with and without product differentiation.}}
\label{fig: user_satisfaction}
\end{figure}
\begin{table}[t]
 \centering
\caption{ {Avg. negotiation time per agent with market size}}
\label{tab: market size}
\begin{tabular}{lcccc} \toprule
Agents & $40$ & $60$ & $80$ & $100$ \\ \midrule
Negotiation time (seconds)&       $4.1437$  & $10.5930$   & $19.9979$ & $43.9713$\\ \bottomrule
 \end{tabular}
 \end{table}
\section{Conclusion}
We have formulated P2P energy trading as an assignment game (coalitional game) over time-varying  communication networks and proposed a novel distributed negotiation algorithm as a clearing mechanism that guarantees stable trading prices in a coalitional game theoretic sense and satisfies the desired economic properties.
The proposed bilateral P2P energy market designs namely, single-contract and multi-contract, encourage prosumers to participate by making P2P trading a favorable choice, considering their economic and social priorities. Furthermore, enabling product differentiation increases user satisfaction and allows for a higher overall market welfare. Finally, the negotiation mechanism via paracontraction operators enables faster convergence to a consensus on a set of bilateral contract prices that represent a competitive equilibrium and belong to the core.

{An interesting extension of our work would be the design of online mechanisms for real-time markets where the core set varies over time, thus accommodating for the short-term uncertainty in RES generation and demand. }

\appendix
To prove the convergence of (\ref{main_it}), as stated in Theorem \ref{theorem: main}, we first provide useful results regarding paracontractions.
\begin{lem}[\cite{Elsner1992}, Thm. 1]\label{lemma: finite family}
Let $\mathcal{T}$ be a finite family of paracontractions such that $\bigcap_{T \in \mathcal{T}} \mathrm{fix}(T)$ $ \neq \varnothing $. Then, the sequence $(\boldsymbol{x}^k)_{k \in \mathbb{N}}$ generated by $\boldsymbol{x}^{k+1} := T^k( \boldsymbol{x}^{k})$ converges to a common fixed-point of the paracontractions that occur infinitely often in the sequence. $\hfill \square$
\end{lem}
\begin{lem}[Doubly stochastic matrix (\cite{Fullmer2018}, Prop. 5 )]\label{lem: Doubly stochastic matrix}
If $W$ is a doubly stochastic matrix then, the linear operator defined by the matrix $W \otimes I_{n}$ under Assumption \ref{asm: graph} is a paracontraction with respect to the mixed vector norm $\|\cdot\|_{2,2}$. $\hfill \square$
\end{lem}

\begin{lem}[Composition of paracontracting operators  (\cite{Fullmer2018}, Prop. 1)]\label{prop:Composition of paracontracting operators }
Suppose $T_1, T_2 : \mathbb{R}^n \rightarrow \mathbb{R}^n$ are paracontractions with respect to same norm $\|\cdot\|$ and $\mathrm{fix}(T_1) \cap \mathrm{fix}(T_2) \neq \varnothing$. Then, the composition $T_1 \circ T_2$ is a paracontraction with respect to the norm $\|\cdot \|$ and $\mathrm{fix}(T_1 \circ T_2) = \mathrm{fix}(T_1) \cap \mathrm{fix}(T_2)$. $\hfill \square$
\end{lem}

\begin{lem}[Stacked vector of paracontractions (\cite{Fullmer2018}, Prop. 4)]\label{prop:stacked paracontraction operator}
Suppose each map $T_{1}, \ldots, T_{m}$ is a paracontraction with respect to $\|\cdot\|_{2}$. Then, the map $\boldsymbol{T}:= \mathrm{col}(T_1, \ldots,  T_N)$ is a paracontraction with respect to $\|\cdot\|_{2,2}$. $\hfill \square$
\end{lem} 
Using these properties, we now show that the sequence of operators generated by the iteration in (\ref{main_it}) is a paracontraction and the set of its fixed-points is a consensus in the intersection of the fixed-point sets of the operators.
\begin{lem}\label{lemma: fix points}
Let $Q$ be the integer in Assumption \ref{asm: Q-con }. Let  $\boldsymbol{T}_1, \ldots, \boldsymbol{T}_Q $ be paracontraction operators with  $  
\bigcap_{r=1}^Q \mathrm{fix}(\boldsymbol{T}_r) =: C$ and let $W_Q  W_{Q-1} \cdots  W_1 $  be the composition of the adjacency matrices where $W_r \in \mathcal{W}$, with $\mathcal{W}$ as in Assumption \ref{asm: fixed graph}. Let $\boldsymbol{W}_r := W_r \otimes I_N $. Then, the composed mapping $ \boldsymbol{x} \mapsto (\boldsymbol{T}_Q\boldsymbol{W}_Q  \circ \cdots \circ  \boldsymbol{T}_1 \boldsymbol{W}_1)(\boldsymbol{x})$ 
\begin{enumerate}[(i)]
    \item is a paracontraction with respect to norm $\| \cdot \|_{2,2}$;
    \item $\mathrm{fix}(\boldsymbol{T}_Q\boldsymbol{W}_Q  \circ \cdots \circ  \boldsymbol{T}_1 \boldsymbol{W}_1) =  \mathcal{A} \cap C$,  
\end{enumerate}
where $\mathcal{A}$ is the consensus set in (\ref{eq: consensus}). $\hfill \square$
\end{lem}

\begin{proof}
(i): It follows directly from Lemmas \ref{lem: Doubly stochastic matrix} and \ref{prop:Composition of paracontracting operators }.\\
(ii): By Lemmas \ref{lem: Doubly stochastic matrix} and \ref{prop:Composition of paracontracting operators }, $\mathrm{fix}(\boldsymbol{T}_Q\boldsymbol{W}_Q  \circ \cdots \circ  \boldsymbol{T}_1 \boldsymbol{W}_1) = \mathrm{fix}(\boldsymbol{T}_Q) \cap \cdots \cap \mathrm{fix}(\boldsymbol{T}_1) \cap \mathrm{fix}(\boldsymbol{W}_Q) \cap \cdots \cap \mathrm{fix}(\boldsymbol{W}_1)$. Again, by Lemmas \ref{lem: Doubly stochastic matrix}, \ref{prop:Composition of paracontracting operators } $\bigcap_{r=1}^Q \mathrm{fix}(\boldsymbol{W}_r) = \mathrm{fix}(\boldsymbol{W}_Q \cdots \boldsymbol{W}_1)$ and since the composition $\boldsymbol{W}_Q \cdots \boldsymbol{W}_1$ is strongly connected, by the Perron-Frobenius theorem, $\mathrm{fix}(\boldsymbol{W}_Q \cdots \boldsymbol{W}_1) = \mathcal{A}$. Furthermore, as $\bigcap_{r=1}^Q \mathrm{fix}(\boldsymbol{T}_r) = C$, $\mathrm{fix}(\boldsymbol{T}_Q\boldsymbol{W}_Q  \circ \cdots \circ  \boldsymbol{T}_1 \boldsymbol{W}_1) = \mathcal{A} \cap C$. 
\end{proof}
With these results, we are now ready to prove Theorem \ref{theorem: main}.\\
\begin{proof}(Theorem \ref{theorem: main})
Let us define the sub-sequence of $\boldsymbol{x}^{k} \text{ for all } k \in \mathbb{N}$ as $\boldsymbol{z}^t = \boldsymbol{x}^{(t-1)Q}$ for each $t \geq 2 $ with $Q$ being the integer in Assumptions \ref{asm: Q-con } and \ref{asm: Q admissible bargaining}. Then,

\begin{equation}\label{eq: z^{k} subsequence}
    \boldsymbol{z}^{t+1} = \boldsymbol{T}^{tQ - 1}\boldsymbol{W}^{tQ - 1}  \circ \cdots \circ  \boldsymbol{T}^{(t-1)Q} \boldsymbol{W}^{(t-1)Q} \boldsymbol{z}^t
\end{equation}
for $t \geq 2$. It follows from Lemma \ref{prop:stacked paracontraction operator} and assertion 1 of Lemma \ref{lemma: fix points} that the maps $\boldsymbol{x} \longmapsto (\boldsymbol{T}^{tQ - 1}\boldsymbol{W}^{tQ - 1}  \circ \cdots \circ  \boldsymbol{T}^{(t-1)Q} \boldsymbol{W}^{(t-1)Q})(\boldsymbol{x}),$ $ t \geq 2 $ are all paracontractions. Also, under Assumption \ref{asm: fixed graph}, there can be only finitely many such maps. Furthermore, by assertion 2 of Lemma \ref{lemma: fix points}, the set of fixed-points of each map is
$\mathcal{X}^*$. Thus, by Lemma \ref{lemma: finite family}, the iteration in (\ref{eq: z^{k} subsequence}) converges to some $ \bar{\boldsymbol{z}} \in \mathcal{X}^*$. 
\end{proof}

\bibliographystyle{IEEEtran}
\bibliography{bibliography.bib}

\begin{thebibliography}{10}
\providecommand{\url}[1]{#1}
\csname url@samestyle\endcsname
\providecommand{\newblock}{\relax}
\providecommand{\bibinfo}[2]{#2}
\providecommand{\BIBentrySTDinterwordspacing}{\spaceskip=0pt\relax}
\providecommand{\BIBentryALTinterwordstretchfactor}{4}
\providecommand{\BIBentryALTinterwordspacing}{\spaceskip=\fontdimen2\font plus
\BIBentryALTinterwordstretchfactor\fontdimen3\font minus
  \fontdimen4\font\relax}
\providecommand{\BIBforeignlanguage}[2]{{%
\expandafter\ifx\csname l@#1\endcsname\relax
\typeout{** WARNING: IEEEtran.bst: No hyphenation pattern has been}%
\typeout{** loaded for the language `#1'. Using the pattern for}%
\typeout{** the default language instead.}%
\else
\language=\csname l@#1\endcsname
\fi
#2}}
\providecommand{\BIBdecl}{\relax}
\BIBdecl

\bibitem{morstyn2018using}
T.~Morstyn, N.~Farrell, S.~J. Darby, and M.~D. McCulloch, ``Using peer-to-peer
  energy-trading platforms to incentivize prosumers to form federated power
  plants,'' \emph{Nature Energy}, vol.~3, no.~2, pp. 94--101, 2018.

\bibitem{han2018incentivizing}
L.~Han, T.~Morstyn, and M.~McCulloch, ``Incentivizing prosumer coalitions with
  energy management using cooperative game theory,'' \emph{IEEE Trans. Power
  Syst.}, vol.~34, no.~1, pp. 303--313, 2018.

\bibitem{Tushar2020}
W.~Tushar, T.~K. Saha, C.~Yuen, D.~Smith, and H.~V. Poor, ``{Peer-to-Peer
  Trading in Electricity Networks: An Overview},'' \emph{IEEE Trans. Smart
  Grid}, vol.~11, no.~4, pp. 3185--3200, 2020.

\bibitem{tushar2018peer}
W.~Tushar, T.~K. Saha, C.~Yuen, P.~Liddell, R.~Bean, and H.~V. Poor,
  ``Peer-to-peer energy trading with sustainable user participation: A game
  theoretic approach,'' \emph{IEEE Access}, vol.~6, pp. 62\,932--62\,943, 2018.

\bibitem{moret2018energy}
F.~Moret and P.~Pinson, ``Energy collectives: a community and fairness based
  approach to future electricity markets,'' \emph{IEEE Trans. Power Syst.},
  vol.~34, no.~5, pp. 3994--4004, 2018.

\bibitem{tushar2018transforming}
W.~Tushar, C.~Yuen, H.~Mohsenian-Rad, T.~Saha, H.~V. Poor, and K.~L. Wood,
  ``Transforming energy networks via peer-to-peer energy trading: The potential
  of game-theoretic approaches,'' \emph{IEEE Signal Processing Magazine},
  vol.~35, no.~4, pp. 90--111, 2018.

\bibitem{morstyn2018multiclass}
T.~Morstyn and M.~D. McCulloch, ``Multiclass energy management for peer-to-peer
  energy trading driven by prosumer preferences,'' \emph{IEEE Trans. Power
  Syst.}, vol.~34, no.~5, pp. 4005--4014, 2018.

\bibitem{Crespo-Vazquez2020}
J.~L. Crespo-Vazquez, T.~A. Skaif, A.~M. Gonzalez-Rueda, and M.~Gibescu, ``{A
  Community-Based Energy Market Design Using Decentralized Decision-Making
  under Uncertainty},'' \emph{IEEE Trans. Smart Grid}, vol. 3053, no.~c, pp.
  1--11, 2020.

\bibitem{Sorin2019}
E.~Sorin, L.~Bobo, and P.~Pinson, ``{Consensus-Based Approach to Peer-to-Peer
  Electricity Markets with Product Differentiation},'' \emph{IEEE Trans. Power
  Syst.}, vol.~34, no.~2, pp. 994--1004, 2019.

\bibitem{Morstyn2019}
T.~Morstyn, A.~Teytelboym, and M.~D. McCulloch, ``{Bilateral contract networks
  for peer-to-peer energy trading},'' \emph{IEEE Trans. Smart Grid}, vol.~10,
  no.~2, pp. 2026--2035, 2019.

\bibitem{nguyen2020optimal}
D.~H. Nguyen, ``Optimal solution analysis and decentralized mechanisms for
  peer-to-peer energy markets,'' \emph{IEEE Trans. Power Syst.}, 2020.

\bibitem{ullah2021two}
M.~H. Ullah and J.-D. Park, ``A two-tier distributed market clearing scheme for
  peer-to-peer energy sharing in smart grid,'' \emph{IEEE Trans. Ind.
  Informat.}, vol.~18, no.~1, pp. 66--76, 2021.

\bibitem{cui2019efficient}
S.~Cui, Y.-W. Wang, Y.~Shi, and J.-W. Xiao, ``An efficient peer-to-peer
  energy-sharing framework for numerous community prosumers,'' \emph{IEEE
  Trans. Ind. Informat.}, vol.~16, no.~12, pp. 7402--7412, 2019.

\bibitem{alskaif2021blockchain}
T.~AlSkaif, J.~L. Crespo-Vazquez, M.~Sekuloski, G.~van Leeuwen, and J.~P.
  Catal{\~a}o, ``Blockchain-based fully peer-to-peer energy trading strategies
  for residential energy systems,'' \emph{IEEE Trans. Ind. Informat.}, vol.~18,
  no.~1, pp. 231--241, 2021.

\bibitem{pradhan2021flexible}
N.~R. Pradhan, A.~P. Singh, N.~Kumar, M.~Hassan, and D.~Roy, ``A flexible
  permission ascription (fpa) based blockchain framework for peer-to-peer
  energy trading with performance evaluation,'' \emph{IEEE Trans. Ind.
  Informat.}, 2021.

\bibitem{schweppe2013spot}
F.~C. Schweppe, M.~C. Caramanis, R.~D. Tabors, and R.~E. Bohn, \emph{Spot
  pricing of electricity}.\hskip 1em plus 0.5em minus 0.4em\relax Springer
  Science \& Business Media, 2013.

\bibitem{Saad2009}
W.~Saad, Z.~Han, M.~Debbah, A.~Hj{\o}rungnes, and T.~Başar, ``{Coalitional
  game theory for communication networks},'' \emph{IEEE Signal Processing
  Magazine}, vol.~26, no.~5, pp. 77--97, 2009.

\bibitem{shapley1971assignment}
L.~S. Shapley and M.~Shubik, ``The assignment game i: The core,''
  \emph{International Journal of game theory}, vol.~1, no.~1, pp. 111--130,
  1971.

\bibitem{tushar2019motivational}
W.~Tushar, T.~K. Saha, C.~Yuen, T.~Morstyn, M.~D. McCulloch, H.~V. Poor, and
  K.~L. Wood, ``A motivational game-theoretic approach for peer-to-peer energy
  trading in the smart grid,'' \emph{App. Ener.}, vol. 243, pp. 10--20, 2019.

\bibitem{tushar2019grid}
W.~Tushar, T.~K. Saha, C.~Yuen, T.~Morstyn, H.~V. Poor, R.~Bean \emph{et~al.},
  ``Grid influenced peer-to-peer energy trading,'' \emph{IEEE Trans. Smart
  Grid}, vol.~11, no.~2, pp. 1407--1418, 2019.

\bibitem{tushar2020coalition}
W.~Tushar, T.~K. Saha, C.~Yuen, M.~I. Azim, T.~Morstyn, H.~V. Poor, D.~Niyato,
  and R.~Bean, ``A coalition formation game framework for peer-to-peer energy
  trading,'' \emph{App. Ener.}, vol. 261, p. 114436, 2020.

\bibitem{luo2018distributed}
F.~Luo, Z.~Y. Dong, G.~Liang, J.~Murata, and Z.~Xu, ``A distributed electricity
  trading system in active distribution networks based on multi-agent coalition
  and blockchain,'' \emph{IEEE Trans. Power Syst.}, vol.~34, no.~5, pp.
  4097--4108, 2018.

\bibitem{yu2014phev}
R.~Yu, J.~Ding, W.~Zhong, Y.~Liu, and S.~Xie, ``Phev charging and discharging
  cooperation in v2g networks: A coalition game approach,'' \emph{IEEE Internet
  of Things Journal}, vol.~1, no.~6, pp. 578--589, 2014.

\bibitem{wang2021blockchain}
Y.~Wang, Z.~Su, J.~Li, N.~Zhang, K.~Zhang, K.-K.~R. Choo, and Y.~Liu,
  ``Blockchain-based secure and cooperative private charging pile sharing
  services for vehicular networks,'' \emph{IEEE Transactions on Vehicular
  Technology}, 2021.

\bibitem{raja2021}
A.~A. Raja and S.~Grammatico, ``Payoff distribution in robust coalitional games
  on time-varying networks,'' \emph{IEEE Trans. Control Netw. Syst.}, early
  access.

\bibitem{bauso2015distributed}
D.~Bauso and G.~Notarstefano, ``Distributed $ n $-player approachability and
  consensus in coalitional games,'' \emph{IEEE Transactions on Automatic
  Control}, vol.~60, no.~11, pp. 3107--3112, 2015.

\bibitem{nedic2017achieving}
A.~Nedic, A.~Olshevsky, and W.~Shi, ``Achieving geometric convergence for
  distributed optimization over time-varying graphs,'' \emph{SIAM Journal on
  Optimization}, vol.~27, no.~4, pp. 2597--2633, 2017.

\bibitem{bauschke2011convex}
H.~H. Bauschke and P.~L. Combettes, \emph{Convex analysis and monotone operator
  theory in Hilbert spaces}, 2nd~ed.\hskip 1em plus 0.5em minus 0.4em\relax
  Springer, 2017.

\bibitem{nedic2013}
A.~Nedić and D.~Bauso, ``Dynamic coalitional tu games: Distributed bargaining
  among players' neighbors,'' \emph{IEEE Transactions on Automatic Control},
  vol.~58, no.~6, pp. 1363--1376, 2013.

\bibitem{Elsner1992}
L.~Elsner, I.~Koltracht, and M.~Neumann, ``{Convergence of sequential and
  asynchronous nonlinear paracontractions},'' \emph{Numer. Math.}, vol.~62,
  no.~1, pp. 305--319, 1992.

\bibitem{Fullmer2018}
D.~Fullmer and A.~S. Morse, ``{A Distributed Algorithm for Computing a Common
  Fixed Point of a Finite Family of Paracontractions},'' \emph{IEEE
  Transactions on Automatic Control}, vol.~63, no.~9, pp. 2833--2843, 2018.

\end{thebibliography}

\end{document}